\documentclass{fundam}

\publyear{2021}
\papernumber{2060}
\volume{182}
\issue{1}

\usepackage{url} 
\usepackage[ruled,lined]{algorithm2e}
\usepackage{graphicx}
\usepackage{tikz}
\usetikzlibrary{automata, positioning, arrows}

\usepackage{petriGamesMacros}
\usepackage[backgroundcolor=cyan!40, linecolor = cyan!50]{todonotes}
\usepackage{hyperref}
\usepackage{svg}
\usepackage{caption}
\usepackage{subcaption}
\usepackage{longtable}

\usepackage{pifont}
\newcommand{\cmark}{\ding{51}}
\newcommand{\xmark}{\ding{55}}
\newcommand{\lightText}[1]{\emph{\color{gray}#1}}
\usepackage{arydshln}

\begin{document}

\title{Canonical Representations for Direct Generation of Strategies in High-level Petri~Games\footnote[3]{This paper is a revised and extended version of \cite{DBLP:conf/apn/GiesekingW21}.}\thanks{
			This work was supported by 
			the German Research Foundation (DFG)
			through the
			Research Training Group 
			SCARE (DFG GRK 1765) 
			and through Grant Petri Games (No. 392735815).}}

\address{Nick W\"urdemann, wuerdemann@informatik.uni-oldenburg.de}

\author{Manuel Gieseking\\
Department of Computing Science,
University of Oldenburg,
Oldenburg, Germany\\
gieseking@informatik.uni-oldenburg.de
\and Nick W\"urdemann\\
Department of Computing Science,
University of Oldenburg,
Oldenburg, Germany\\
wuerdemann@informatik.uni-oldenburg.de } 

\maketitle

\runninghead{M. Gieseking, N. W\"urdemann}{Canon.\ Representations for Direct Gen.\ of Strategies in High-level Petri~Games}
\vspace*{-1cm}\begin{abstract}
	Petri games are a multi-player game model
	for the synthesis of distributed systems with multiple concurrent processes
	based on Petri nets.
	The processes are the players in the game represented by the token of the net.
	The players are divided into two teams:
	the controllable system and the uncontrollable environment.
	An individual controller is synthesized for each process based only
	on their locally available causality-based information.
	For one environment player and 
	a bounded number of system players,
	the problem of solving Petri games can be reduced 
	to that of solving B\"uchi games.
  
  High-level Petri games are  a concise 
  representation of ordinary Petri games.
  Symmetries, derived from a high-level representation,
  can be exploited to significantly
  reduce the state space in the corresponding B\"uchi game.
  We present a new construction for solving high-level Petri games. It involves
  the definition of a unique, canonical representation of the reduced B\"uchi game.
  This allows us to translate a strategy 
  in the B\"uchi game directly into a strategy in the Petri game.
  An implementation applied on six structurally different benchmark families
  shows in almost all cases a performance increase for larger state spaces.
\end{abstract}
\section{Introduction}\label{sec:Introduction}
Whether telecommunication networks,
electronic banking, or the world wide web,
\emph{distributed systems} 
are all around us and are becoming
increasingly more widespread.
Though an entire system may appear as one unit,
the local controllers in a network often act autonomously on only incomplete information
to avoid constant communication.
These independent agents must behave correctly under all 
possible uncontrollable behavior of the environment.
\emph{Synthesis}~\cite{Church/57/Applications}
avoids the error-prone task of manually implementing such local controllers
by automatically generating correct ones 
from a given specification 
(or stating the nonexistence of such controllers).
In case of a single process in the underlying model,
synthesis approaches have been successfully applied in nontrivial applications
(e.g., \cite{DBLP:conf/date/BloemGJPPW07}, 
\cite{DBLP:journals/trob/Kress-GazitFP09}). 
Due to the incomplete information in systems with
multiple processes progressing on their individual rate,
modeling \emph{asynchronous distributed systems}
is even more cumbersome and particularly benefits from a synthesis approach.

\emph{Petri games} \cite{DBLP:journals/corr/FinkbeinerO14}
(based on an underlying Petri net \cite{DBLP:books/sp/Reisig85a} where the tokens are the players in the game)
are a well-suited multi-player game model
for the \emph{synthesis of asynchronous distributed systems}
because of its subclasses with comparably low complexity results.
For Petri games with a single \emph{environment} (\emph{uncontrollable}) player,
a bounded number of \emph{system} (\emph{controllable}) players,
and a \emph{safety objective}, i.e.,
all players have to avoid designated \emph{bad places},
deciding the existence of a winning strategy for the system players is \textsc{exptime}-complete~\cite{DBLP:journals/iandc/FinkbeinerO17}.
This problem is called the \emph{realizability problem}.
The result is obtained via a reduction to a two-player B\"uchi game 
with enriched markings, so called \emph{decision sets}, as states.

\emph{High-level Petri nets} \cite{DBLP:series/eatcs/Jensen92} 
can concisely model large distributed systems.
Correspondingly,
\emph{high-level Petri games} \cite{DBLP:journals/corr/abs-1904-05621}
are a concise high-level representation of ordinary Petri games.
For solving high-level Petri games,
the \emph{symmetries} \cite{10.5555/115220.115224} of the system 
can be exploited to build a \emph{symbolic B\"uchi game}
with a significantly smaller number of states \cite{DBLP:journals/acta/GiesekingOW20}.
The states are \emph{equivalence classes} of decision sets and called \emph{symbolic decision sets}.
For generating a Petri game strategy in a high-level Petri game
the approach proposed in \cite{DBLP:journals/acta/GiesekingOW20}
resorts to the original strategy construction in \cite{DBLP:journals/iandc/FinkbeinerO17},
i.e., the equivalence classes of a symbolic two-player strategy are dissolved
and a strategy for the standard two-player game is generated.
Figure~\ref{fig:HLandLLDiagram} shows in the two bottom layers the relation of the elements just described.

\begin{figure}[!t]
	\centering	
	\begin{tikzpicture}[transform shape,node distance=\ydis and \xdis, on grid, scale=1.0]
	\renewcommand{\xdis}{4.5cm}
	\renewcommand{\ydis}{2.85cm}
	\node[align=center](HLPG){High-level\\ Petri game $ \SG $};
	\node[below=of HLPG, align=center] (LLPG) {Petri game\\ $ \PG=\HLtoLL(\SG) $};
	\node[right=of HLPG, align=center](HL2PG){Symbolic\\two-player\\game $ \STPG $};
	\node[below=of HL2PG, align=center] (LL2PG) {Two-player\\game $ \TPG $};
	\node[right=of HL2PG, align=center](HL2PGS){ Strategy\\ in $ \STPG $ };
	\node[below=of HL2PGS, align=center] (LL2PGS) {Strategy\\in $ \TPG $};
	\node[right=of HL2PGS](HLPGS){};
	\node[below=of HLPGS, align=center] (LLPGS) {Strategy $ \PGS $\\ in $ \PG $};
	\draw[->]
	([xshift=-3pt]HLPG.south) edge node [above, rotate=90] {transform} 
	node [right,xshift=16pt] {\cite{DBLP:journals/corr/abs-1904-05621}} ([xshift=-3pt]LLPG.north)
	(LLPG) edge node [above] {reduce} node [below] {\cite{DBLP:journals/iandc/FinkbeinerO17}} (LL2PG)
	(LL2PG) edge node [above] {solve} node[below]{}(LL2PGS)
	([xshift=3pt]LLPG.north) -- node [below, rotate=90, align=center] {represent} ([xshift=3pt]HLPG.south);
	\draw 
	(HL2PG)-- node[above, rotate=90] {bisimilar} node [right] {\cite{DBLP:journals/acta/GiesekingOW20}} (LL2PG);
	\draw[->]
	(HLPG) edge node [above] {reduce} node [below] {\cite{DBLP:journals/acta/GiesekingOW20}} (HL2PG)
	(HL2PG) edge node [above] {solve}node[below]{} (HL2PGS)
	(HL2PGS) edge node [above, rotate=90] {dissolve} node [below, rotate=90] {symmetries} node [right,xshift=5mm] {\cite{DBLP:journals/acta/GiesekingOW20}} (LL2PGS)
	(LL2PGS) edge node [above] {translate} node [below] {\cite{DBLP:journals/iandc/FinkbeinerO17}} (LLPGS)
	;
	
	\node[above=of HL2PG,align=center,yshift=-1/3*\ydis](HLC){Symbolic game\\ with canon. rep.\\  $ \STPGc $};
	\node[right=of HLC,align=center](HLCS){Strategy\\ in $ \STPGc $};
	
	\draw[->,very thick]
	(HLPG) edge node [above,sloped] {reduce} node [below,sloped] {\begin{NoHyper}\refSection{CanonRepsOfSymbDecsets},\ref{sec:DirStratGen}\end{NoHyper}} (HLC.west)
	(HLC) edge node [above] {solve}node[below]{} (HLCS)
	(HLCS) edge node [above,sloped] {generate} node[below,sloped]{\begin{NoHyper}\refSection{DirStratGen}\end{NoHyper}} (LLPGS)
	;
	\node at ($ (HLPG)!0.5!(HLC.west) $) [xshift=8.0mm,yshift=1.5mm] {\hyperlink{target:CanonRepsOfSymbDecsets}{\phantom{$ 3 $}}};
	\node at ($ (HLPG)!0.5!(HLC.west) $) [xshift=10.5mm,yshift=3mm] {\hyperlink{target:DirStratGen}{\phantom{$ 4 $}}};
	\node at ($ (HLCS)!0.5!(LLPGS) $) [xshift=0.8mm,yshift=-4.8mm] {\hyperlink{target:DirStratGen}{\phantom{$ A $}}};
	
	\draw[draw opacity=0]
	(HLC) edge node [rotate=90] {$ \simeq $} (HL2PG)
	(HLCS) edge node [rotate=90] {$ \simeq $} (HL2PGS)
	;
	\end{tikzpicture}%
%
	\caption{
		An overview of the scope of this paper.
		The connections between the different elements 
		describe their interplay and where these
		methods are introduced. 
		The connections labeled with ``solve''
		mean that a two-player (B\"uchi) game can be solved 
		by standard algorithms in game theory (e.g., \cite{DBLP:conf/dagstuhl/2001automata}).
		The bottom level corresponds to 
		the original reduction in \cite{DBLP:journals/iandc/FinkbeinerO17},
		the level above corresponds to the high-level counterparts described in 
		\cite{DBLP:journals/corr/abs-1904-05621,DBLP:journals/acta/GiesekingOW20},
		and the top level contains the elements introduced in this~paper.
		The new reduction is marked by thick edges.
	}%
	\label{fig:HLandLLDiagram}%
\end{figure}
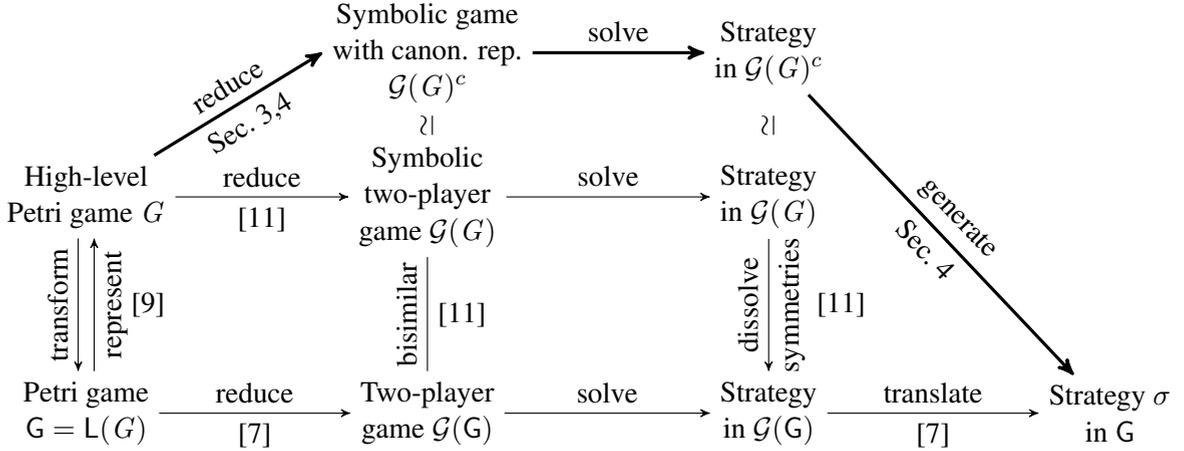
In this paper, we propose a new construction for solving high-level Petri games 
to avoid this detour while generating the strategy.
In~\cite{DBLP:journals/acta/GiesekingOW20} the symbolic B\"uchi game is generated by
comparing each newly added state with all already added ones for equivalence, i.e.,
the \emph{orbit problem} 
must be answered.
The new approach calculates a \emph{canonical representation} for each newly added state 
(the \emph{constructive orbit problem} \cite{DBLP:conf/cav/ClarkeEJS98}),
and only stores these representations.
This generation of a symbolic B\"uchi game 
with canonical representations is based on the corresponding ideas for reachability graphs
from \cite{DBLP:journals/tcs/ChiolaDFH97}.
As in \cite{DBLP:journals/acta/GiesekingOW20},
we consider safe Petri games with a high-level representation,
and exclude Petri games where the system can proceed infinitely without the environment.
For the decidability result we consider, as in \cite{DBLP:journals/iandc/FinkbeinerO17}, 
Petri games with only one environment player, i.e., in every reachable marking there is at most one token on an environment place.

One of the main advantages of the new approach is that the canonical representations
allow to directly generate a Petri game strategy
from a symbolic B\"uchi game strategy
without explicitly resolving all symmetries (cp. thick edges and the top level in \refFig{HLandLLDiagram}).
\begin{figure}[!t]
	\centering
	\begin{subfigure}[t]{0.34\textwidth}
		\centering
		\begin{tikzpicture}[node distance=\ydis and \xdis,on grid,>=stealth',bend angle=45,scale=0.95, transform shape]
	\renewcommand{\xdis}{18mm}
	\renewcommand{\ydis}{10.2mm}
	
	\node[envplace]		(Env)	[label=right:{$ \mathit{Env} $},tokens=1]{};
	\node[transition]	(i)		[below=of Env, label={right:$ \mathit{d} $}]{};
	\node[envplace]		(I)		[below=of i, label={left:$ \mathit{I} $}]{};
	\node[transition]	(inf)	[below=of I, label={[label distance=-2pt,yshift=-2pt]left:$ \mathit{inf} $}]{};
	\node[envplace]		(R)		[below=of inf,yshift=-\ydis, label={[label distance=-1pt]left:$ \mathit{R} $}]{};
	\node[transition]	(h)		[below=of R, yshift=-2mm, label={left:$ \mathit{h} $}]{};
	\node[sysplace]		(H)		[below=of h, label={left:$ \mathit{H} $}]{};
	\node[sysplace]		(Sys)	[right=of inf, label={right:$ \mathit{Sys} $}]{};
	\node[transition]	(a)		[below=of Sys, label={right:$ \mathit{a} $}]{};
	\node[sysplace]		(A)		[below=of a, label={right:$ \mathit{A} $}]{};
	\node[transition]	(b)		[below=of A, yshift=-2mm, label={right:$ \mathit{b} $}]{};
	\node[sysplace,bad]	(B)		[below=of b, label={right:$ \mathit{B} $}]{};
	\node[scale=0.85] at (Sys) 		(SysNotation){$ \basecl_1 $};

	\path[->]
	(i) 
	edge[pre] node[left,scale=0.9]{$ \bullet $}	(Env)
	edge[post] node[left,scale=0.9]{$ x $}	(I)
	(inf)
	edge[pre] node[left,scale=0.9]{$ x $}	(I)
	edge[pre,bend left]	node[above,inner sep=1pt,scale=0.9]{$ \basecl_1 $}	(Sys)
	edge[post] node[left,scale=0.9]{$ x $}	(R)
	edge[post,bend right] node[below,inner sep=2pt,scale=0.9]{$ \basecl_1 $}	(Sys)
	(a)
	edge[pre] node[right,scale=0.9]{$ y $} (Sys)
	edge[post] node[right,scale=0.9]{$ (y,x) $} (A)
	(h)
	edge[pre] node[left,pos=0.7,scale=0.9]{$ x $} (R)
	edge[pre] node[above,sloped,inner sep=1pt,scale=0.8,pos=0.57]{$ \basecl_1\!\times\!\{x\} $} (A)
	edge[post] node[left,scale=0.9]{$ x $} (H)
	(b)
	edge[pre] node[right,scale=0.9]{$ (y,x) $} (A)
	edge[post] node[right,scale=0.9]{$ y $}(B)
	;
	
	\node at (i) (description) [xshift=24mm, yshift=-0.23*\ydis, align=left,draw=black,scale=0.9]
	{ $ \basecl_1=\{ c_1,c_2,c_3 \} $\\
		$ \basecl_2=\{ \bullet \} $\\
		$ x,y\in\basecl_1 $};		
\end{tikzpicture}%
		\subcaption{High-level Petri game~$ \SG $}%
		\label{fig:symmetricGameExampleHLPG}
	\end{subfigure}%
	~
	\begin{subfigure}[t]{0.65\textwidth}
		\centering
		\begin{tikzpicture}[node distance=\ydis and \xdis,on grid,>=stealth',bend angle=45,scale=0.95, transform shape]
	\renewcommand{\xdis}{34mm}
	\renewcommand{\ydis}{8mm}
	\newcommand{\smallxshift}{8mm}
	
	\node[envplace]		(Env)	[label={[inner sep=0pt,xshift=4mm,yshift=0.7mm]below:$ \mathit{Env} $},tokens=1]{};
	\node[transition]	(i2)		[below=of Env, label={right:$ \mathit{d}_2 $}]{};
	\node[envplace]		(I2)		[below=of i2,yshift=0.5mm, label={right:$ \mathit{I}_2 $}]{};	
	\node[transition]	(inf2)		[below=of I2,yshift=0.5mm, label={[xshift=0.5mm]right:$ \mathit{inf}_2 $}]{};
	\node[sysplace]		(Sys2)		[below=of inf2, yshift=-3mm, label={[xshift=0.5mm]right:$ \mathit{Sys}_2 $},tokens=1]{};	
	\node[transition]	(a22)		[below=of Sys2, label={[inner sep=1pt]below:$ \mathit{a}_{22} $}]{};	
	\node[transition] at (a22)	(a21)		[xshift=-\smallxshift, label={[inner sep=1pt]below:$ \mathit{a}_{21} $}]{};
	\node[transition] at (a22)	(a23)		[xshift=\smallxshift, label={[inner sep=1pt]below:$ \mathit{a}_{23} $}]{};
	\node[sysplace]		(A22)		[below=of a22, yshift=-4mm, label={[inner sep=0pt,xshift=-3mm]above:$ \mathit{A}_{22} $}]{};
	\node[sysplace] at (A22)	(A21)		[xshift=-\smallxshift, label={[inner sep=0pt,xshift=-3mm]above:$ \mathit{A}_{21} $}]{};	
	\node[sysplace] at (A22)	(A23)		[xshift=\smallxshift, label={[inner sep=0pt,xshift=-3mm]above:$ \mathit{A}_{23} $}]{};
	\node[envplace]	at (A21)	(R2)		[xshift=-8mm, label={[inner sep=0pt,xshift=-3mm]above:$ \mathit{R}_2 $}]{};
	\node[transition]	(h2)		[below=of R2, yshift=-4mm, label={[inner sep=0pt]left:$ \mathit{h}_2 $}]{};
	\node[sysplace]		(H2)		[below=of h2, yshift=-3mm, label={[inner sep=0pt]right:$ \mathit{H}_{2} $}]{};
	
	\node[transition]	(b22)		[below=of A22, yshift=-4mm, label={[inner sep=0pt]below:$ \mathit{b}_{22} $}]{};	
	\node[transition] at (b22)	(b21)		[xshift=-\smallxshift, label={[inner sep=0pt]below:$ \mathit{b}_{21} $}]{};
	\node[transition] at (b22)	(b23)		[xshift=\smallxshift, label={[inner sep=0pt]below:$ \mathit{b}_{23} $}]{};
	\node[sysplace,bad]		(B22)		[below=of b22, yshift=-3mm, label={[inner sep=1pt]right:$ \mathit{B}_{2} $}]{};
		
	\node[transition]	(i1)		[left=of i2, label={left:$ \mathit{d}_1 $}]{};
	\node[envplace]		(I1)		[below=of i1,yshift=0.5mm, label={left:$ \mathit{I}_1 $}]{};	
	\node[transition]	(inf1)		[below=of I1,yshift=0.5mm,  label={[yshift=0.5mm]right:$ \mathit{inf}_1 $}]{};
	\node[sysplace]		(Sys1)		[below=of inf1, yshift=-3mm, label={left:$ \mathit{Sys}_1 $},tokens=1]{};	
	\node[transition]	(a12)		[below=of Sys1, label={[inner sep=1pt]below:$ \mathit{a}_{12} $}]{};	
	\node[transition] at (a12)	(a11)		[xshift=-\smallxshift, label={[inner sep=1pt]below:$ \mathit{a}_{11} $}]{};
	\node[transition] at (a12)	(a13)		[xshift=\smallxshift, label={[inner sep=1pt]below:$ \mathit{a}_{13} $}]{};
	\node[sysplace]		(A12)		[below=of a12, yshift=-4mm, label={[inner sep=0pt,xshift=-3mm]above:$ \mathit{A}_{12} $}]{};
	\node[sysplace] at (A12)	(A11)		[xshift=-\smallxshift, label={[inner sep=0pt,xshift=-3mm]above:$ \mathit{A}_{11} $}]{};	
	\node[sysplace] at (A12)	(A13)		[xshift=\smallxshift, label={[inner sep=0pt,xshift=-3mm]above:$ \mathit{A}_{13} $}]{};	
	\node[envplace]	at (A11)	(R1)		[xshift=-8mm, label={[inner sep=0pt,xshift=-3mm]above:$ \mathit{R}_1 $}]{};
	\node[transition]	(h1)		[below=of R1, yshift=-4mm, label={[inner sep=0pt]left:$ \mathit{h}_1 $}]{};
	\node[sysplace]		(H1)		[below=of h1, yshift=-3mm, label={[inner sep=0pt]right:$ \mathit{H}_{1} $}]{};
	
	\node[transition]	(b12)		[below=of A12, yshift=-4mm, label={[inner sep=0pt]below:$ \mathit{b}_{12} $}]{};	
	\node[transition] at (b12)	(b11)		[xshift=-\smallxshift, label={[inner sep=0pt]below:$ \mathit{b}_{11} $}]{};
	\node[transition] at (b12)	(b13)		[xshift=\smallxshift, label={[inner sep=0pt]below:$ \mathit{b}_{13} $}]{};
	\node[sysplace,bad]		(B12)		[below=of b12, yshift=-3mm, label={[inner sep=1pt]right:$ \mathit{B}_{1} $}]{};
	
	\node[transition]	(i3)		[right=of i2, label={right:$ \mathit{d}_3 $}]{};
	\node[envplace]		(I3)		[below=of i3,yshift=0.5mm, label={right:$ \mathit{I}_3 $}]{};	
	\node[transition]	(inf3)		[below=of I3,yshift=0.5mm,  label={right:$ \mathit{inf}_3 $}]{};
	\node[sysplace]		(Sys3)		[below=of inf3, yshift=-3mm, label={right:$ \mathit{Sys}_3 $},tokens=1]{};	
	\node[transition]	(a32)		[below=of Sys3, label={[inner sep=1pt]below:$ \mathit{a}_{32} $}]{};	
	\node[transition] at (a32)	(a31)		[xshift=-\smallxshift, label={[inner sep=1pt]below:$ \mathit{a}_{31} $}]{};
	\node[transition] at (a32)	(a33)		[xshift=\smallxshift, label={[inner sep=1pt]below:$ \mathit{a}_{33} $}]{};
	\node[sysplace]		(A32)		[below=of a32, yshift=-4mm, label={[inner sep=0pt,xshift=-3mm]above:$ \mathit{A}_{32} $}]{};
	\node[sysplace] at (A32)	(A31)		[xshift=-\smallxshift, label={[inner sep=0pt,xshift=-3mm]above:$ \mathit{A}_{31} $}]{};	
	\node[sysplace] at (A32)	(A33)		[xshift=\smallxshift, label={[inner sep=0pt,xshift=-3mm]above:$ \mathit{A}_{33} $}]{};
	\node[envplace]	at (A31)	(R3)		[xshift=-8mm, label={[inner sep=0pt,xshift=-3mm]above:$ \mathit{R}_3 $}]{};
	\node[transition]	(h3)		[below=of R3, yshift=-4mm, label={[inner sep=0pt,yshift=-1mm]left:$ \mathit{h}_3 $}]{};
	\node[sysplace]		(H3)		[below=of h3, yshift=-3mm, label={[inner sep=0pt]right:$ \mathit{H}_{3} $}]{};
	
	\node[transition]	(b32)		[below=of A32, yshift=-4mm, label={[inner sep=0pt]below:$ \mathit{b}_{32} $}]{};	
	\node[transition] at (b32)	(b31)		[xshift=-\smallxshift, label={[inner sep=0pt]below:$ \mathit{b}_{31} $}]{};
	\node[transition] at (b32)	(b33)		[xshift=\smallxshift, label={[inner sep=0pt]below:$ \mathit{b}_{33} $}]{};
	\node[sysplace,bad]		(B32)		[below=of b32, yshift=-3mm, label={[inner sep=1pt]right:$ \mathit{B}_{3} $}]{};
	
	\begin{pgfonlayer}{bg}
	\path[->, gray!70]
	(inf1)	edge[pre]	(I1)
			edge[pre and post] (Sys1)
			edge[pre and post] (Sys2)
			edge[pre and post] (Sys3)
	(inf2)	edge[pre]	(I2)
			edge[pre and post] (Sys1)
			edge[pre and post] (Sys2)
			edge[pre and post] (Sys3)
	(inf3)	edge[pre]	(I3)
			edge[pre and post] (Sys1)
			edge[pre and post] (Sys2)
			edge[pre and post] (Sys3)
	(a11) 	edge[pre]	(Sys1)
			edge[post] 	(A11)
	(a12) 	edge[pre]	(Sys1)
			edge[post] 	(A12)
	(a13) 	edge[pre]	(Sys1)
			edge[post] 	(A13)
	(a21) 	edge[pre]	(Sys2)
			edge[post] 	(A21)
	(a22) 	edge[pre]	(Sys2)
			edge[post] 	(A22)
	(a23) 	edge[pre]	(Sys2)
			edge[post] 	(A23)
	(a31) 	edge[pre]	(Sys3)
			edge[post] 	(A31)
	(a32) 	edge[pre]	(Sys3)
			edge[post] 	(A32)
	(a33) 	edge[pre]	(Sys3)
			edge[post] 	(A33)
	(b11) 	edge[pre]	(A11)
			edge[post] 	(B12)
	(b12) 	edge[pre]	(A12)
			edge[post] 	(B12)
	(b13) 	edge[pre]	(A13)
			edge[post] 	(B12)
	(b21) 	edge[pre]	(A21)
			edge[post] 	(B22)
	(b22) 	edge[pre]	(A22)
			edge[post] 	(B22)
	(b23) 	edge[pre]	(A23)
			edge[post] 	(B22)
	(b31) 	edge[pre]	(A31)
			edge[post] 	(B32)
	(b32) 	edge[pre]	(A32)
			edge[post] 	(B32)
	(b33) 	edge[pre]	(A33)
			edge[post] 	(B32)
	;
	\end{pgfonlayer}
	\path[->]
	(h1)	edge[pre]	(R1)
			edge[pre]	(A11)
			edge[pre,in=210]	(A21)
			edge[pre,in=210,out=30]	(A31)
			edge[post]	(H1)
	(h2)	edge[pre]	(R2)
			edge[pre]	(A12)
			edge[pre]	(A22)
			edge[pre,in=225,out=25]	(A32)
			edge[post]	(H2)
	(h3)	edge[pre]	(R3)
			edge[pre]	(A13)
			edge[pre]	(A23)
			edge[pre]	(A33)
			edge[post]	(H3)
	(i1)	edge[pre,]	(Env)
			edge[post]	(I1)
	(i2)	edge[pre]	(Env)
			edge[post]	(I2)
	(i3)	edge[pre]	(Env)
			edge[post]	(I3)
	;
	\node at (R1|-inf1)	(H1){};
	\node at (R2|-inf2)	(H2){};
	\node at (R3|-inf3)	(H3){};
	\path[draw,->,rounded corners]
	(inf1)--(H1.center)--(R1);
	\path[draw,->,rounded corners]
	(inf2)--(H2.center)--(R2);
	\path[draw,->,rounded corners]
	(inf3)--(H3.center)--(R3);
\end{tikzpicture}%
		\subcaption{The represented P/T Petri game~$ \PG = \HLtoLL(\SG) $.}%
		\label{fig:motivation_update}
	\end{subfigure}%
	\caption{
		\label{fig:symmetricGameExample}		
		\emph{Client/Server}: The environment decides on one out of three computers
		to host a server.
		The system players (computers)
		can win the game by getting informed
		on the decision of the environment
		and connecting correctly.
	}
\end{figure}
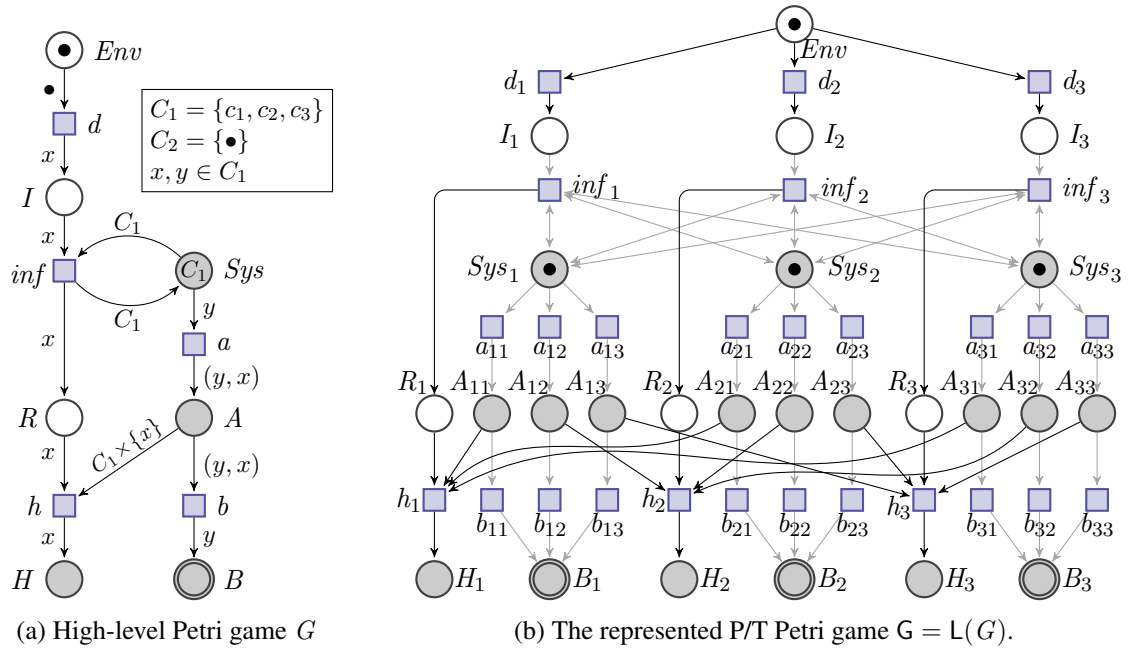
Another advantage is the complexity for constructing the symbolic B\"uchi game.
Even though the calculation of the canonical representation comes with a fixed cost, 
less comparisons can be necessary,
depending on the input system.
We implemented the new algorithm and applied our tool on the benchmark families
used in \cite{DBLP:journals/acta/GiesekingOW20} and Example~\ref{ex:example}.
The results show in general a performance increase with an
increasing number of states for almost all benchmark families.

We now introduce the example
on which we demonstrate the successive development stages of the 
presented techniques throughout the paper.
\begin{example}
	\label{ex:example}
	The high-level Petri game $ \SG $ depicted in \refFig{symmetricGameExampleHLPG} 
	models a simplified scenario
	where one out of three computers 
	must host a server for the others to connect to. 
	The environment nondeterministically decides which computer 
	must be the host.
	The places in the net are partitioned into \emph{system places} (gray)
	and \emph{environment places} (white). 
	An object on a place is a player in the corresponding team.
	\emph{Bad places} are double-bordered.
	The variables $ x,y $ on arcs are bound only locally to the transitions,
	and an assignment of objects to these variables is called a \emph{mode}
	of the transition.
	
	The environment player $ \bullet $,
	initially residing on place $ \mathit{Env} $,
	decides via transition~$ d $ in mode $ x=\tilde{c} $
	on a computer~$ \tilde{c} $
	that should host the server.
	The system players
	(computers~$ c_1,c_2,c_3\in\basecl_1 $),
	initially residing on place~$ \mathit{Sys} $,
	can either directly, individually connect themselves 
	to another computer $ c' $ via transition~$ \mathit{a} $ in mode~$ [y=c_i,x=\tilde{c}] $, 
	or wait for transition~$ \mathit{inf} $
	to be enabled.
	When they choose to connect themselves directly,
	after firing transition~$ \mathit{a} $ in different modes,
	the corresponding pairs of computers 
	reside on place~$ \mathit{A} $. 
	Since the players always have to give the possibility to proceed in the game, 
	and transition $ \mathit{h} $ cannot get enabled any more,
	they must take transition~$ \mathit{b} $
	to the bad place~$ \mathit{B} $.
	So instead, all players should initially
	only allow transition~$ \mathit{inf} $
	(in every possible mode).
	After the decision of the environment, 
	transition~$ \mathit{inf}$ can be fired in mode $x=\tilde{c} $, 
	placing $ \tilde{c} $ on $ \mathit{R} $. 
	In this firing, the system players get informed
	on the environment's decision.
	Back on place~$ \mathit{Sys} $ they can,
	equipped with this knowledge, 
	each connect to the 
	computer~$ \tilde{c} $ via transition~$ \mathit{a} $,
	putting the three objects
	$ (c_1,\tilde{c}),(c_2,\tilde{c}),(c_3,\tilde{c}) $ 
	on place~$ \mathit{A} $.
	Thus, transition~$ \mathit{h} $
	can be fired in mode~$ x=\tilde{c} $, and the 
	game terminates with $ \tilde{c} $ in $ \mathit{H} $. Since the
	system players avoided
	reaching the bad place~$ \mathit{B} $,
	they win the play.
	This scenario is highly symmetric,
	since it does not matter which computer is chosen to be the host,
	as long as the others connect themselves~correctly.
\end{example}
The remainder of this paper is structured as follows:
In \refSection{PetriNetsAndPetriGames} we recall 
the definitions of 
(high-level) Petri nets and (high-level) Petri games.
In \refSection{CanonRepsOfSymbDecsets} 
we present the idea, formalization,
and construction of canonical representations.
In \refSection{DirStratGen} we show the application 
of these canonical representations in
the symbolic two-player B\"uchi game,
and how to directly generate a Petri game strategy.
In \refSection{ExperimentalResults}, experimental results of the presented techniques are shown.
Section~\ref{sec:RelatedWork} presents the related work and
\refSection{Conclusions} concludes the paper.

\section{Petri Nets and Petri Games}\label{sec:PetriNetsAndPetriGames}

This section recalls (high-level) Petri nets and (high-level) Petri games,
and the associated concept of strategies
established in \cite{DBLP:journals/iandc/FinkbeinerO17,DBLP:journals/corr/abs-1904-05621,DBLP:journals/acta/GiesekingOW20}.
To provide an easier and more intuitive approach to the subject,
the full formal definition of the strategy in a Petri game
can be found in App.~\ref{app:strategyFormal}.
Figure~\ref{fig:symmetricGameExample} serves as an illustration.

\subsection{P/T Petri Nets}
A (marked P/T) \emph{Petri net} is a tuple 
$ \PN=(\places,\transitions,\flowfunc,\marking_\mathsf{0}) $,
with the disjoint sets of \emph{places}~$ \places $
and \emph{transitions}~$ \transitions $,
a \emph{flow function} $ \flowfunc:(\places\times\transitions)\cup(\transitions\times\places)\to \N $,
and an \emph{initial marking} $ \marking_\mathsf{0} $,
where a \emph{marking} is a multi-set $ \marking:\places\to\N $ 
that indicates the number of tokens on each place.
$ \flowfunc(\mathsf{x},\mathsf{y})=n>0 $ means there is an \emph{arc} of \emph{weight} $ n $
from node $ \mathsf{x} $ to $ \mathsf{y} $ describing the flow of tokens in the net.
A transition $ \mathsf{t}\in\transitions $ is \emph{enabled} in a marking $ \marking $
if $ \forall \place\in\places: \flowfunc(\place,\transition)\leq\marking(\place) $.
If~$ \transition $ is enabled then $ \transition $ can \emph{fire} in $ \marking $,
leading to a new marking $ \marking' $ calculated by
$\forall \place\in\places: \marking'(\place)=\marking(\place)-\flowfunc(\place,\transition)+\flowfunc(\transition,\place) $.
This is denoted by $ \marking[\transition\rangle \marking' $.
$ \PN $ is called \emph{safe} if for all markings~$ \marking $ that can be reached from $ \marking_0 $
by firing a sequence of transitions we have $ \forall \place\in\places: \marking(\place)\leq 1 $.
For each transition $ \transition\in\transitions $ we define the \emph{pre-} and \emph{postset} of $ \mathsf{t} $ as the multi-sets $ \preset{\mathsf{t}}  = \flowfunc (\cdot, \mathsf{t})$
and $ \postset{\mathsf{t}}  = \flowfunc (\mathsf{t}, \cdot) $ over~$ \places $.

An example for a Petri net can be seen in \refFig{motivation_update}.
Ignoring the different shades and potential double borders for now, 
the net's places are depicted as circles and its transitions as squares.
Dots represent the number of tokens on each place in the initial marking of the net.
The flow is depicted as weighted arcs between places and transitions. 
Missing weights are interpreted as arcs of weight~$ 1 $.
In the initial marking of the Petri game depicted in \refFig{motivation_update},
all transitions $ \mathit{a}_{ij} $ and $ \mathit{d}_i $ are enabled.
Firing, e.g., $ \mathit{d}_1 $ results in the marking with one token on 
$ \mathit{I}_1 $, $ \mathit{Sys}_1 $, $ \mathit{Sys}_2 $, and $ \mathit{Sys}_3 $, each.

\subsection{P/T Petri Games}
Petri games are an extension of Petri nets to incomplete information games between two teams of players:
the controllable system vs. the uncontrollable environment.
The tokens on places in a Petri net represent the individual players.
The place a player resides on determines their team membership.
Particularly, a player can switch teams.
For that, the places are divided into system places and environment places.
A play of the game is a concurrent execution of transitions in the net.
During a play, the \emph{knowledge} of each player is represented by their \emph{causal history},
i.e., all visited places and used transitions to reach to current place.
Players enrich this local knowledge when synchronizing in a joint transition.
There, the complete knowledge of all participating players are exchanged.
Based on this, players allow or forbid transitions in their postset.
A transition can only fire if every player in its preset allows the execution.
The system players in a Petri game win a play if they satisfy a safety-condition,
given by a designated set of bad places they must not reach.

Formally, a (P/T) \emph{Petri game} is a tuple
$ \PG=(\sysplaces,\envplaces,\transitions,\flowfunc,\marking_\mathsf{0},\badplaces) $,
with a set of \emph{system places} $ \sysplaces $,
a set of \emph{environment places} $ \envplaces $,
and a set of \emph{bad places} $ \badplaces\subseteq\sysplaces $.
The set of all places is denoted by $ \places=\sysplaces\dot{\cup}\envplaces $,
and $ \transitions,\flowfunc,\marking_0 $ are the remaining components of a
Petri net $ \PN=(\places,\transitions,\flowfunc,\marking_\mathsf{0}) $,
called the \emph{underlying net} of $ \PG $.
We consider only Petri games with finitely many places and transitions.

In \refFig{motivation_update}, a Petri game is depicted.
We just introduced the underlying net of the game.
The system places are shaded gray, the environment places are white. 
Bad places are marked by a double border.
This Petri game is the P/T-version of the high-level Petri game described in the introduction.
The three tokens/system players residing on $ \mathit{Sys}_i $
represent the computers.
The environment player residing on $ \mathit{Env} $ makes their decision 
which computer should host a server by taking a transition $ \mathit{d}_i $.
The system players can then get informed of the decision and react accordingly 
as described above.

A \emph{strategy} for the system players 
in a Petri game $ \PG $
can be formally expressed 
as a \emph{subprocess} of the
\emph{unfolding} \cite{DBLP:series/eatcs/EsparzaH08}.
Formal definitions of the concepts are presented in App.~\ref{app:strategyFormal}.
In the unfolding of a Petri net,
every loop is unrolled and every backward
branching place is expanded by 
duplicating the place, 
so that every transition
represents the unique occurrence 
of a transition during an execution of the net.
The causal dependencies in $ \PG $ 
(and thus, the knowledge of the players)
are naturally represented in its unfolding,
which is the unfolding of the 
underlying net with system-,
environment-, and bad places 
marked correspondingly.

A strategy is obtained from an unfolding 
by deleting some of the branches 
that are under control of the system players. 
This subprocess has to meet three
conditions: 
(i) The strategy must be \emph{deadlock-free}, 
to avoid trivial solutions;
it must allow the possibility to continue, 
whenever the system can proceed.
Otherwise the system players could win with the respect to 
the safety objective (bad places) if they decide to do nothing.
(ii) The system players must act in a \emph{deterministic} way,
i.e., in no reachable marking of the strategy two
transitions involving the same system player are enabled.
(iii) \emph{Justified refusal:}
if a transition is not in the strategy,
then the reason is that a system player 
in its preset forbids all occurrences 
of this transition in the strategy.
Thus, no pure environment decisions are restricted,
and system players can only allow or forbid a
transition of the original net, based on only their knowledge. 
In a \emph{winning} strategy, the system players cannot reach bad places.

\begin{figure}[!h]
	\centering
	\input{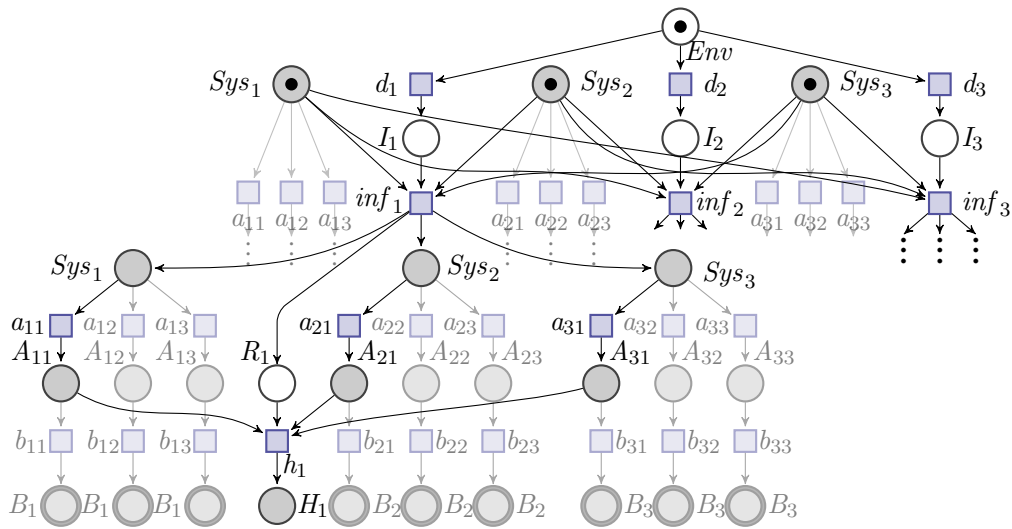}%
	\caption{\label{fig:UnfAndStrat} 
		Part of a winning strategy for the system players in $ \PG $ (solid),
		obtained by deleting some of the branches of the unfolding (solid and greyed out).
	}
\end{figure}

In \refFig{UnfAndStrat}, we see the 
already informally described
winning strategy for the system players
in the Petri game $ \PG $.
For clarity, we only show the case in which the environment chose 
the first computer to be the host completely. 
All computers, \emph{after} getting informed of the environment's
decision, act correspondingly and connect to the first computer.
The remaining branches in the unfolding are cut off in the strategy.
The other two cases (after firing $ \mathit{inf}_2 $ or $ \mathit{inf}_3 $) are analogous.
The formal
definitions of unfoldings and strategies can be found in App.~\ref{app:strategyFormal}.

\subsection{High-Level Petri Nets}
While in P/T Petri nets only tokens can reside on places,
in high-level Petri nets each place is equipped with a \emph{type}
that describes the form of data (also called \emph{colors}) the place can hold.
Instead of weights, each arc between a place $ p $ and a transition $ t $ is
equipped with an \emph{expression}, indicating which of these colors are taken from
or laid on $ p $ when firing $ t $.
Additionally, each transition $ t $ is equipped with a \emph{guard}
that restricts when $ t $ can fire.

Formally, a \emph{high-level} Petri net is a tuple
$ \SN=(\hlplaces,\hltransitions,\hlflowfunc,\type,\expression,\guard,\hlmarking_0) $,
with a set of \emph{places} $ \hlplaces $,
a set of \emph{transitions} $ \hltransitions $ satisfying $ \hlplaces\cap\hltransitions=\emptyset $,
a \emph{flow relation} $ \hlflowfunc\subseteq(\hlplaces\times\hltransitions)\cup(\hltransitions\times\hlplaces) $,
a \emph{type function} $ \type $ from $ \hlplaces $
such that for each place $ p $, $ \type(p) $ is the set of colors that can lie on $ p $,
a mapping~$ \expression $ that, for every transition~$ t $, assigns to each arc~$ (p,t) $ (or $ (t,p) $)
in $ \hlflowfunc $ an expression $ \expression(p,t) $
(or $ \expression(t,p) $) indicating which colors are
withdrawn from~$ p $ (or laid on $ p $) when $ t $ is fired,
a \emph{guard function} $ \guard $ that equips each transition~$ t $
with a Boolean expression $ \guard(t) $,
an \emph{initial marking}~$ \hlmarking_0 $,
where a \emph{marking} in $ \SN $ is a function~$ \hlmarking $ with domain $ \hlplaces $ 
indicating what colors reside on each place,
i.e., $ \forall p\in\hlplaces:
\hlmarking(p)\in[\type(p)\to\N] $.

\refFig{symmetricGameExampleHLPG} a high-level Petri net is depicted.
As in the P/T case, we ignore the different shadings and borders of places for now. 
The types of the places can be deducted from the surrounding arcs. 
For example, the place $ \mathit{E} $ has the type $ \type(\mathit{E}) = \basecl_2 = \{ \bullet \} $,
and the place $ \mathit{A} $ has the type $ \type(\mathit{A}) = \basecl_1\times\basecl_1 $.
Each arc is equipped with an expression, e.g., $ \expression(\mathit{Sys},\mathit{a})=y $,
and $ \expression(\mathit{a},\mathit{A})=(y,x) $.
In the given net, all guards of transitions are \emph{true} and therefore not depicted.

Typically, expressions and guards will contain variables.
A \emph{mode} (or valuation) $ v $ of a transition~$ t\in\hltransitions $ 
assigns to each variable $ x $
occurring in $ \guard(t) $, or in
any expression $ \expression(p,t) $ or $ \expression(t,p) $,
a value $ v(x) $.
The set $ \valuations(t) $ contains all modes of $ t $.
Each $ v\in\valuations(t) $ assigns a Boolean value, denoted by $ v(t) $, to~$ g(t) $,
and to each arc expression $ e(p,t) $ or $ e(t,p) $ a multi-set over $ \type(p) $,
denoted by $ v(p,t) $ or~$ v(t,p) $.
A transition $ t $ is \emph{enabled in a mode} $ v\in\valuations(t) $
in a marking~$ \hlmarking $
if $ v(t)=\true $ and for each arc $ (p,t)\in\hlflowfunc $ and every $ c\in\type(p) $
we have $ v(p,t)(c)\leq \hlmarking(p)(c) $.
The marking $ \hlmarking' $ reached by firing~$ t $ in mode $ v $ from $ \hlmarking $
(denoted by $ \hlmarking[t.v\rangle\hlmarking' $) is calculated by
$ \forall p\in\places\ \forall c\in\type(p) : 
\hlmarking'(p)(c)=\hlmarking(p)(c)-v(p,t)(c)+v(t,p)(c) $.

A high-level Petri net $ \SN $ can be \emph{transformed} 
into a P/T Petri net $ \HLtoLL(\SN) $ with
$ \places=\{ p.c\with p\in\hlplaces, c\in\type(p) \} $,
$ \transitions=\{ t.v\with t\in\hltransitions, v\in\valuations(t), v(t)=\true \} $,
the flow~$ \flowfunc $ defined by 
$ \forall p.c\in\places\ \forall t.v\in\transitions :
\flowfunc(p.c,t.v)=v(p,t)(c)\land \flowfunc(t.v,p.c)=v(t,p)(c) $,
and initial marking $ \marking_\mathsf{0} $ defined by 
$ \forall p.c\in\places: \marking_\mathsf{0}(p.c)=\hlmarking_0(p)(c) $.
The two nets then have the same semantics:
the number of tokens on a place~$ p.c $ in a marking in~$ \HLtoLL(\SN) $ indicates
the number of colors~$ c $ on place~$ p $ 
in the corresponding marking in $ \SN $.
Firing a transition~$ t.v $ in $ \HLtoLL(\SN) $
corresponds to firing transition~$ t $ in mode~$ v $ in $ \SN $.
We say a high-level Petri net $ \SN $ \emph{represents} the P/T Petri net~$ \HLtoLL(\SN) $.

\subsection{High-Level Petri Games}
Just as P/T Petri games are structurally based on P/T Petri nets,
high-level Petri games are based on high-level Petri nets.
Again, in a \emph{high-level Petri game}
$ \SG=(\hlsysplaces,\hlenvplaces,\hltransitions,\hlflowfunc,\type,\expression,\guard,\hlmarking_0,\hlbadplaces) $
with an \emph{underlying high-level net} 
$ \SN=(\hlplaces,\hltransitions,\hlflowfunc,\type,\expression,\guard,\hlmarking_0) $,
the places~$ \hlplaces $ are divided into 
\emph{system places}~$ \hlsysplaces $ and
\emph{environment places}~$ \hlenvplaces $.
The set $ \hlbadplaces\subseteq\hlsysplaces $ indicates the bad places.
High-level Petri games represent P/T Petri games: 
a high-level Petri game $ \SG $ 
(with underlying high-level net $ \SN $)
\emph{represents} a
P/T Petri game $ \HLtoLL(\SG) $ with underlying P/T~Petri net $ \HLtoLL(\SN) $.
The classification of places $ p.c $ in $ \HLtoLL(\SG) $
into system-, environment-, and bad places
corresponds to the places $ p $ in the high-level game.

In \refFig{symmetricGameExample}, a high-level Petri game $ \SG $
and its represented Petri game $ \PG=\HLtoLL(\SG) $ are depicted.
For the sake of clarity, we abbreviated the nodes in $ \HLtoLL(\SG) $.
Thus, e.g., the transition $ \mathit{a}.[x=c_1, y=c_2] $
is renamed to $ \mathit{a}_{12} $.
We often use notation from the represented P/T Petri game to express situations in a high-level game.

\hypertarget{target:CanonRepsOfSymbDecsets}{}%
\section{Canonical Representations of Symbolic Decision Sets}\label{sec:CanonRepsOfSymbDecsets}
In this paper, we investigate for a given high-level Petri game $ \SG $ with one environment player
whether the system players in the corresponding P/T Petri game $ \HLtoLL(\SG) $ have a winning strategy (and possibly generate one).
This problem is solved via a reduction to a symbolic two-player B\"uchi game \(\STPGc\).
The general idea of this reduction is, as in~\cite{DBLP:journals/iandc/FinkbeinerO17},
to equip the markings of the Petri game with a set of transitions for each system player (called \emph{commitment sets})
which allows the players to fix their next move.
In the generated B\"uchi game, only a subset of all interleavings is taken into account.
The selection arises from delaying 
the moves of the environment player until no system player can progress
without interacting with the environment.
By that, each system player gets informed about the environment's last position
during their next move.
This means that in every state, every system player 
knows the current position of the environment or learns it in the next step,
before determining their next move.
Thus, the system players can be considered to be completely informed about the game.
This is only possible due to the existence of only one environment player.
For more environment players such interleavings would not
ensure that each system player is informed (or gets informed in their next move)
about all environment positions.
The nodes of the game are called \emph{decision sets}.
In \cite{DBLP:journals/acta/GiesekingOW20}, 
symmetries in the Petri net are exploited to define
equivalence classes of decision sets, called \emph{symbolic decision sets}.
These are used to create an equivalent, but significantly smaller, B\"uchi~game.

In this section we introduce the new canonical representations of symbolic decision sets
which serve as nodes for the new B\"uchi game. 
We transfer relations between and properties of (symbolic) decision sets to the established canonical representations.
We start by recalling the definitions of symmetries in Petri nets~\cite{10.5555/115220.115224}
and of (symbolic) decision sets \cite{DBLP:journals/acta/GiesekingOW20}.

From now on we consider high-level Petri games $ \SG $
representing a safe P/T Petri game \(\HLtoLL(\SG)\) 
that has one environment player,
a bounded number of system players with a safety objective,
and where the system cannot proceed infinitely without the environment.
We denote the class of all such high-level Petri games by $ \SGclass $.

\subsection{Symmetric Nets}
High-level representations are often created using \emph{symmetries}~\cite{10.5555/115220.115224} in a Petri net.
Conversely,
in some high-level nets, symmetries can be read directly from the given specification.
A class of nets which allow this are the so called \emph{symmetric nets} (SN) \cite{ChiolaDFH93,Chiola1991}.\footnote{Symmetric Nets were formerly known as Well-Formed Nets (WNs). The renaming was part of the ISO standardization~\cite{DBLP:conf/forte/HillahKPT06}.}
In symmetric nets, the types of places
are selected from given (finite) \emph{basic color classes} $ \basecl_1,\dots,\basecl_n $.
For every place $ p\in\hlplaces $, we have 
$ \type(p)=\basecl_1^{\iota_p(1)}\!\times\!\cdots\!\times\!\basecl_n^{\iota_p(n)} $
for natural numbers $ \iota_p(1),\dots,\iota_p(n)\in\N $,
where $ \basecl_i^{x} $ denotes the $ x $-fold Cartesian product of $ \basecl_i $.\footnote{
	In the Cartesian products $ \type(p) $ and $ \valuations(t) $, we omit all $ \basecl_i^{x} $ with $ x=0 $ (empty sets).}
The possible values of variables contained in guards and arc expressions
are also basic classes. 
Thus, the modes of each transition~\mbox{$ t\in\hltransitions $}
are also given by a Cartesian product 
$ \valuations(t)=\basecl_1^{\iota_t(1)}\!\times\!\cdots\!\times\!\basecl_n^{\iota_t(n)} $.
A basic color class may be \emph{ordered}, i.e., 
equipped with a successor function $ \succfunc $
satisfying $ \forall c\in \basecl_{i}: \succfunc^{|\basecl_i|}(c)=c \land \forall 1\leq k< |\basecl_i|:  \succfunc^{k}(c)\neq c $,
where $ \succfunc^{k} $ describes the $ k $-fold successive execution of $ \succfunc $.
Additionally, each basic color class $ \basecl_i $ is possibly
partitioned into \emph{static subclasses} $ \basecl_i=\bigcup_{q=1}^{n_i}\basecl_{i,q} $.
Guards and arc expressions treat all elements in a static subclass equally.

\begin{example}
	The underlying high-level net $ \SN $ in Fig.~\ref{fig:symmetricGameExample}
	is a symmetric net with basic color classes $ \basecl_1=\{ c_1,c_2,c_3 \} $
	and $ \basecl_2=\{ \bullet \} $.
	We have, e.g., $ \type(A)=\basecl_1\!\times\basecl_1$ (i.e., $ \iota_A(1)=2,\iota_A(2)=0 $) and $\valuations(a)=\basecl_1\!\times\basecl_1 $
	(the two coordinates representing $ y $ and $ x $),
	and therefore, $ \iota_a(1)=2,\iota_a(2)=0 $.
	Both basic color classes $ \basecl_1 $ and $ \basecl_2 $ are not ordered and not further partitioned into static subclasses, i.e., 
	$ \basecl_1=\bigcup_{q=1}^{1}\basecl_{1,q}=\basecl_{1,1} $ and $ \basecl_2=\bigcup_{q=1}^{1}\basecl_{2,q}=\basecl_{2,1} $, and therefore $ n_1=n_2=1 $.
\end{example}

\begin{proposition}[\cite{ChiolaDFH93}]\label{prop:wlogSN}
	Any high-level Petri net can be transformed into a SN with the same basic structure,
	same place types,
	and equivalent arc labeling.
\end{proposition}
The \emph{symmetries} $ \symmetries_\SN $ in a symmetric net $ \SN $
are all tuples $ s=(s_1,\dots,s_n) $ such that each $ s_i $
is a permutation on $ \basecl_i $ satisfying $ \forall q=1,\dots,n_i : s_i(\basecl_{i,q})=\basecl_{i,q} $,
i.e., it respects the partition into static subclasses. In case of an ordered basic class we only consider rotations w.r.t.\ the successor function. This means that if $ \basecl_{i} $ is ordered and divided into two or more static subclasses then the only possibility for $ s_i $ is the identity $ \identity_{\basecl_i} $. 
A symmetry $ s $ can be applied to a single color $ c\in\basecl_i $
by $ s(c)=s_i(c) $.
The application to tuples, e.g., colors on places or transition modes,
is defined by the application in each entry.
The set $ \symmetries_\SN $, 
together with the function composition $ \circ $, 
forms a group
with identity $ (\identity_{\basecl_i})_{i=1}^n $.
In the represented P/T Petri net $ \HLtoLL(\SN) $, 
symmetries can be applied to places 
$ \place=p.c\in\places $ and transitions~$ \transition=t.v\in\transitions $
by defining $ s(p.c)=p.s(c) $ and $ s(t.v)=t.s(v)$.
The structure of symmetric nets ensures
$ \forall s\in\symmetries_\SN\ \forall \transition\in\transitions:
\preset{s(\transition)}=s(\preset{\transition}) \land \postset{s(\transition)}=s(\postset{\transition}) $, and
that symmetries are compatible with 
the firing relation, i.e., $ \forall s\in\symmetries_\SN : 
\marking[\transition\rangle\marking' \Leftrightarrow
s(\marking)[s(\transition)\rangle s(\marking') $.
In a symmetric net, we can w.l.o.g. assume the initial marking $ \hlmarking_0 $
to be \emph{symmetric}, i.e., $ \forall s\in\symmetries_\SN :s(\hlmarking_0)=\hlmarking_0 $.

\subsection{Symbolic Decision Sets}
Consider a high-level Petri game $ \SG $ with its corresponding P/T Petri game representation
$ \HLtoLL(\SG)=(\sysplaces,\envplaces,\transitions,\flowfunc,\marking_\mathsf{0},\badplaces) $.
A \emph{decision set} is a set
$ \decset\subseteq \places\!\times\!(\powerset(\transitions)\cup\top) $.
An element $ (\place,\comset)\in\decset $
with $ \comset\subseteq\postset{\place} $
indicates there is a player on place~$ \place $
who allows all transitions in $ \comset $ to fire.
$ \comset $ is then called a \emph{commitment set}.
An element $ (\place,\top)\in\decset $ indicates the player on place~$ \place $
has to 
choose a commitment set in the next step.
The step of this decision is called \emph{$\top $-resolution}.

In a $ \top $-resolution, each $ \top $-symbol in a decision set $ \decset $ is replaced with a suitable commitment set.
This relation is denoted by $ \decset[\top\rangle \decset' $.
If there are no $ \top $-symbols in $ \decset $,
a transition $ \transition $ is \emph{enabled}, if
$ \forall \place\in\preset{\transition}\ \exists (\place,\comset)\in\decset : \transition\in\comset $,
i.e., there is a token on every place in $ \preset{\transition} $ (as for markings) and additionally, 
$ \transition $ is in every commitment set of such a token.
In the process of firing an enabled transition, the tokens are moved according to the flow~$ \flowfunc $.
The moved or generated tokens on system places are then equipped with a $ \top $-symbol,
while the tokens on environment places allow \emph{all} transitions 
in their postset.
This relation is denoted by $ \decset[\transition\rangle\decset' $.
The \emph{initial decision set} of the game is given by
$ 
\decset_0=
\{ (\place,\{\transition\in\transitions\with \place\in\preset{\transition} \}) \with \place\in\envplaces\cap\marking_{\mathsf{0}} \}
\cup
\{ (\place,\top) \with \place\in\sysplaces\cap\marking_{\mathsf{0}} \}
$,
i.e., the environment in the initial marking allows all possible transitions, 
the system players have to choose a commitment set.
\begin{example}\label{ex:DecisionFirings}
	Assume in the Petri game in Fig.~\ref{fig:symmetricGameExampleHLPG} that
	the computers initially allow transition~$ \mathit{inf} $ in every mode.
	The environment player on $ \mathit{Env} $ fires transition~$ \mathit{d} $ in mode~$ c_1 $.
	After that, the system gets informed of the environment's decision 
	via transition~$ \mathit{inf}$ in mode~$c_1 $.
	The system players, now back on $ \mathit{Sys} $,
	decide that they all want to assign themselves to $ c_1 $ via a $ \top $-resolution.
	This corresponds to the following sequence of decision sets, where we abbreviate $ \mathit{Sys} $ by $ \mathit{S} $.
	\begin{center}
		\begin{tikzpicture}[scale = 0.75, transform shape]\renewcommand{\xdis}{7cm}
		\tikzstyle{decisionset} = [rectangle,fill=lightgray,
		draw,align=center,minimum size=7mm]
		\tikzstyle{env} = [fill=white]
		\tikzstyle{sys} = [rounded corners]
		\node[decisionset,env]	(D0)	{
			$ \decset_0  $};
		\node[decisionset,env] at (D0) [xshift=3.5cm]	(D1)	{
			$ \left( \mathit{Env}.\bullet, \{ \mathit{d}.c_1,\mathit{d}.c_2,\mathit{d}.c_3 \} \right) $\\
			$ \left( \mathit{S}.c_1, \{ \mathit{inf}\!.c_1,\mathit{inf}\!.c_2,\mathit{inf}\!.c_3 \} \right)  $\\
			$ \left( \mathit{S}.c_2, \{ \mathit{inf}\!.c_1,\mathit{inf}\!.c_2,\mathit{inf}\!.c_3 \} \right)  $\\
			$ \left( \mathit{S}.c_3, \{ \mathit{inf}\!.c_1,\mathit{inf}\!.c_2,\mathit{inf}\!.c_3 \} \right)  $
		};
		\node[decisionset,env] at (D1) [xshift=5.8cm]	(D2)	{
			$ \left( \mathit{I}.c_1, \{ \mathit{inf}\!.c_1\} \right) $\\
			$ \left( \mathit{S}.c_1, \{ \mathit{inf}\!.c_1,\mathit{inf}\!.c_2,\mathit{inf}\!.c_3 \} \right)  $\\
			$ \left( \mathit{S}.c_2, \{ \mathit{inf}\!.c_1,\mathit{inf}\!.c_2,\mathit{inf}\!.c_3 \} \right)  $\\
			$ \left( \mathit{S}.c_3, \{ \mathit{inf}\!.c_1,\mathit{inf}\!.c_2,\mathit{inf}\!.c_3 \} \right)  $
		};
		\node[decisionset,env] at (D2) [xshift=4.9cm]	(D3)	{
			$ \left( \mathit{R}.c_1, \{ g.c_1\} \right) $\\
			$ \left( \mathit{S}.c_1, \top \right)  $\\
			$ \left( \mathit{S}.c_2, \top \right)  $\\
			$ \left( \mathit{S}.c_3, \top \right)  $
		};
		\node[decisionset,env] at (D3) [xshift=3.7cm]	(D4)	{
			$ \left( \mathit{R}.c_1, \{ g.c_1\} \right) $\\
			$ \left( \mathit{S}.c_1, \{a.(c_1,c_1)\} \right)  $\\
			$ \left( \mathit{S}.c_2, \{a.(c_2,c_1)\} \right)  $\\
			$ \left( \mathit{S}.c_3, \{a.(c_3,c_1)\} \right)  $
		};
		\draw[->]
		(D0) edge node[above]{$ \top $} (D1)
		(D1) edge node[above]{$ \mathit{d}.c_1 $} (D2)
		(D2) edge node[above]{$ \mathit{inf}.c_1 $} (D3)
		(D3) edge node[above]{$ \top $} (D4)
		;
		\end{tikzpicture}
	\end{center}
\end{example} 

Each decision set describes a situation in the Petri game. 
In the two-player B\"uchi game used to solve the Petri game \cite{DBLP:journals/iandc/FinkbeinerO17},
the decision sets constitute the nodes.
A high-level Petri game $ \SG $ has the same symmetries $ \symmetries_\SN $
as its underlying symmetric net $ \SN $.
They can be applied to decision sets
by applying them to every occurring color $ c $ or mode $ v $.
For a decision set $ \decset $, an equivalence class
$ \{ s(\decset)\with s\in\symmetries_\SN \} $
is called the \emph{symbolic decision set} of $ \decset $,
and contains symmetric situations in the Petri game.
In \cite{DBLP:journals/acta/GiesekingOW20}, these equivalence classes replace 
individual decision sets in the two-player B\"uchi game
to achieve a substantial state space reduction.

\begin{example}
	Consider the second to last decision set in the sequence above.
	This situation is symmetric to the cases where the environment 
	chose computer $ c_2 $ or $ c_3 $ as the host.
	In the example $ \SG $, we have the two color classes $ \basecl_1 $ and $ \basecl_2 $.
	Since $ |\basecl_2|=1 $, the only permutation on $ \basecl_2 $ is $ \identity_{\basecl_2} $.
	Thus, the symmetries in $ \SG $ are the permutations on $ \basecl_1 $.
	Symmetries transform the elements in the symbolic decision set into each other.
	The corresponding symbolic decision set contains the following three elements $ \decset$, $\decset' $, and $ \decset'' $:
	\begin{center}
		\begin{tikzpicture}[scale = 0.75, transform shape]\renewcommand{\xdis}{5cm}
		\tikzstyle{decisionset} = [rectangle,fill=lightgray,
		draw,align=center,minimum size=7mm]
		\tikzstyle{env} = [fill=white]
		\tikzstyle{sys} = [rounded corners]
		\node[decisionset,env](D4)	{
			$ \left( \mathit{R}.c_1, \{ g.c_1\} \right) $\\
			$ \left( \mathit{S}.c_1, \top \right)  $\\
			$ \left( \mathit{S}.c_2, \top \right)  $\\
			$ \left( \mathit{S}.c_3, \top \right)  $
		};
		\node[decisionset,env] at (D4) [xshift=-7cm] (D4')	{
			$ \left( \mathit{R}.c_2, \{ g.c_2\} \right) $\\
			$ \left( \mathit{S}.c_1, \top \right)  $\\
			$ \left( \mathit{S}.c_2, \top \right)  $\\
			$ \left( \mathit{S}.c_3, \top \right)  $
		};
		\node[decisionset,env] at (D4) [xshift=7cm] (D4'')	{
			$ \left( \mathit{R}.c_3, \{ g.c_3\} \right) $\\
			$ \left( \mathit{S}.c_1, \top \right)  $\\
			$ \left( \mathit{S}.c_2, \top \right)  $\\
			$ \left( \mathit{S}.c_3, \top \right)  $
		};
		\path[->]
		(D4) edge[in=20,out=10,loop] node[right] {$ \mathit{id}_{\basecl_1} $} (D4)
		(D4) edge[in=160,out=170,loop] node[left] {$ c_2\leftrightarrow c_3 $} (D4)
		(D4) edge node[above,pos=0.6] {$ c_1\leftrightarrow c_2 $} (D4')
		([yshift=-3mm]D4.west) edge node[below] {$ c_1\mapsto c_2\mapsto c_3\mapsto c_1 $} ([yshift=-3mm]D4'.east)
		(D4) edge node[above,pos=0.6] {$ c_1\leftrightarrow c_3 $} (D4'')
		([yshift=-3mm]D4.east) edge node[below] {$ c_1\mapsto c_3\mapsto c_2\mapsto c_1 $} ([yshift=-3mm]D4''.west)
		;
		
		\node at (D4) [xshift = 1.45cm, yshift=-0.65cm](D4name) {\large$ \decset $};
		\node at (D4') [xshift = 1.5cm, yshift=0.65cm](D4'name) {\large$ \decset' $};
		\node at (D4'') [xshift = -1.55cm, yshift=0.65cm](D4''name) {\large$ \decset'' $};
		\end{tikzpicture}
	\end{center}
Each edge between two decision sets corresponds to the application of a symmetry. 
The abbreviated notation $ c\mapsto c' \mapsto c'' \mapsto c $ means that 
each element is mapped to the next in line. Analogously, $ c \leftrightarrow c' $
means that $ c $ and $ c' $ are switched.
\end{example}

\subsection{Canonical Representations}
In order to exploit symmetries to reduce the size of the state space, 
one aims to consider only one representative of each of the equivalence classes
induced by the symmetries.
This can be done either by checking whether a newly generated state
is equivalent to any already generated one, or
by transforming each newly generated state into an equivalent,
canonical representative.
In \cite{DBLP:journals/acta/GiesekingOW20} we consider the former approach. 
The nodes of the symbolic B\"uchi game are symbolic decision sets.
In the construction, an arbitrary representative $ \overline{\decset} $
is chosen for each of these equivalence classes. 
This means, when reaching a new node~$ \decset' $, 
we must apply every symmetry $ s\in\symmetries_\SN  $ to test whether there 
already is a representative $ \overline{\decset'}=s(\decset') $,
or whether $ \decset' $ is in a new symbolic decision set.

In this section, we now aim at the second approach and
define the new \emph{canonical representations} of symbolic decision sets.
For that, we first define \emph{dynamic representations}, 
and then show how to construct a canonical one.
We use these instead of (arbitrary representatives of) 
symbolic decision sets in the construction of 
the symbolic B\"uchi game in \refSection{DirStratGen}.

\subsubsection{Dynamic Representations}
A dynamic representation is an abstract description of a symbolic decision set.
It consists of 
dynamic subclasses of variables, a function linking these dynamic subclasses to static subclasses of the net, and a dynamic decision set
where the dynamic subclasses replace explicit colors.
Any (valid) assignment of values to the variables in the dynamic 
subclasses results in a decision set in the equivalence class.

Formally, a \emph{dynamic representation}
is a tuple $ \rep=(\dynsubclasses,\assdynsubcl,\decsetfct) $,
with the set of \emph{dynamic subclasses} 
$ \dynsubclasses=\{ \dynsubcl_{i}^j\with 1\leq i \leq n, 1\leq j \leq m_{i} \} $
for natural numbers $ m_{i} $,
a function $ \assdynsubcl: \dynsubclasses\to\N $, and 
a \emph{dynamic decision set}~$ \decsetfct $.
A dynamic subclass $ \dynsubcl_{i}^j $ contains a finite number $ |\dynsubcl_{i}^j| $ of variables with values in $ \basecl_{i,q} $ where
$ q=\assdynsubcl(\dynsubcl_{i}^j)\in\{1,\dots,n_i \} $, i.e., $ \assdynsubcl $ describes from which dynamic subclass these variables are. 
The function satisfies $ j< k\Rightarrow \assdynsubcl(\dynsubcl_{i}^j)< \assdynsubcl(\dynsubcl_{i}^k)$,
i.e., the dynamic subclasses are grouped by assigned static subclass.
In total, there are as many variables with values in $ \basecl_{i,q} $ as there are colors, i.e., $ \sum_{\assdynsubcl(\dynsubcl_{i}^j)=q} |\dynsubcl_{i}^j|=|\basecl_{i,q}| $.
A dynamic decision set is the same as a decision set,
with dynamic subclasses replacing explicit colors.

The components $ \dynsubclasses $ and $ \assdynsubcl $ described above are schematically depicted in Fig.~\ref{fig-schematic-depiction}
for a generic example. We see two basic color classes $ \basecl_{1} $ (not divided into static subclasses, i.e., $ \basecl_{1}=\basecl_{1,1} $) and $ \basecl_{2} $ (divided into the three static subclasses $ \basecl_{2,1},\basecl_{2,2},\basecl_{2,3} $).
The dynamic subclasses in the top level are given by $ \dynsubclasses=\{ \dynsubcl_1^1,\dynsubcl_1^2,\dynsubcl_2^1,\dots, \dynsubcl_2^6 \} $.
The function $ \assdynsubcl $ is indicated by the wave arrows,
e.g., $ \assdynsubcl(\dynsubcl_2^1)=\assdynsubcl(\dynsubcl_2^2)=1 $ or $ \assdynsubcl(\dynsubcl_2^4)=2 $.

An \emph{assignment} $ \validassignment:\bigcup_{i=1}^n \basecl_i\to \dynsubclasses $ is \emph{valid} if it respects the cardinality of dynamic subclasses, i.e., 
$ {| \{ c\in\basecl_{i} \with \validassignment(c)=\dynsubcl_{i}^j \} |}=|\dynsubcl_{i}^j| $,
as well as the function $ \assdynsubcl $, i.e., $ \validassignment(c)=\dynsubcl_i^j\land \assdynsubcl(\dynsubcl_i^j)=q\Rightarrow c\in C_{i,q}  $.
In case of an ordered static subclass $ \basecl_i $,
consecutive elements (w.r.t.\ $ \succfunc $) must be mapped to consecutive dynamic subclasses,
i.e., $ \validassignment(c) = \dynsubcl_{i}^j \Rightarrow \validassignment(\succfunc(c))\in\{\dynsubcl_{i}^j, \dynsubcl_{i}^{j+1}\} $.
Every valid assignment of colors $ c\in\basecl_{i} $ to the $ \dynsubcl_{i}^j$, \mbox{$1\leq j\leq m_{i} $}, gives a partition of~$ \basecl_{i} $.

For every decision set $ \decset $ in a symbolic decision set 
with dynamic representation $ \rep=(\dynsubclasses,\assdynsubcl,\decsetfct) $,
there is a valid assignment $ \validassignment_\decset $
such that $ \decset=\validassignment_\decset^{-1}(\decsetfct) $.
In general, there are several dynamic representations of a symbolic decision~set.
On the other hand, every decision set reached from the dynamic decision set in such a way 
is in the same symbolic decision set.
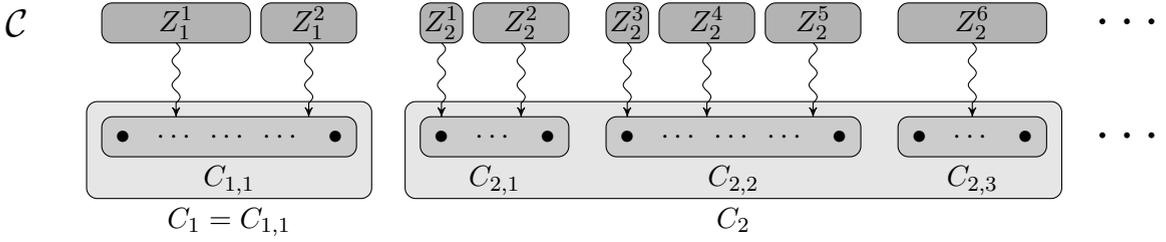
\begin{figure}[t]
	\centering
	\begin{tikzpicture}
	\renewcommand{\xdis}{7mm}
	\renewcommand{\ydis}{15mm}
		\node[] (c11) {\phantom{\scalebox{1.5}{$ \bullet $}}};
		\node[right=of c11] (c12) {$\cdots$};
		\node[right=of c12,xshift=\xdis] (c1n-1) {$\cdots$};
		\node[right=of c1n-1] (c1n) {\phantom{\scalebox{1.5}{$ \bullet $}}};
		\node at ($ (c12)!0.5!(c1n-1) $) (dotsC1) {$ \cdots $};
		
		\node[right=of c1n,xshift=\xdis] (c21) {\phantom{\scalebox{1.5}{$ \bullet $}}};
		\node[right=of c21,xshift=\xdis] (c22) {\phantom{\scalebox{1.5}{$ \bullet $}}};
		\node at ($ (c21)!0.5!(c22) $) (dotsC21) {$ \cdots $};
		\node[right=of c22,xshift=0.5*\xdis] (c23) {\phantom{\scalebox{1.5}{$ \bullet $}}};
		\node[right=of c23] (c231) {$ \cdots $};
		\node[right=of c231,xshift=\xdis] (c24) {$ \cdots $};
		\node[right=of c24] (c241) {\phantom{\scalebox{1.5}{$ \bullet $}}};
		\node at ($ (c231)!0.5!(c24) $) (dotsC22) {$ \cdots $};
		\node[right=of c241,xshift=0.5*\xdis] (c25) {\phantom{\scalebox{1.5}{$ \bullet $}}};
		\node[right=of c25,xshift=\xdis] (c26) {\phantom{\scalebox{1.5}{$ \bullet $}}};
		\node at ($ (c25)!0.5!(c26) $) (dotsC23) {$ \cdots $};
		
		\node at (c11) {$ \bullet $};
		\node at (c1n) {$ \bullet $};
		\node at (c21) {$ \bullet $};
		\node at (c22) {$ \bullet $};
		\node at (c23) {$ \bullet $};
		\node at (c241) {$ \bullet $};
		\node at (c25) {$ \bullet $};
		\node at (c26) {$ \bullet $};
		
		\node[yshift=-2.5mm,anchor=north] at (dotsC1) (C1) {$ \basecl_{1,1} $};
		\node[yshift=-2.5mm,anchor=north] at (dotsC21) (C21) {$ \basecl_{2,1} $};
		\node[yshift=-2.5mm,anchor=north] at (dotsC22) (C22) {$ \basecl_{2,2} $};
		\node[yshift=-2.5mm,anchor=north] at (dotsC23) (C23) {$ \basecl_{2,3} $};

		\node[yshift=-8mm,anchor=north] at (dotsC1) (C1) {$ \basecl_1=\basecl_{1,1} $};
		\node[yshift=-8mm,anchor=north] at (dotsC22) (C22) {$ \basecl_{2} $};
		
		\begin{pgfonlayer}{bg}
		\draw[rounded corners, fill=gray!20] 
			([xshift=-2mm,yshift=2mm]c11.north west) rectangle ([xshift=2mm,yshift=-5.5mm]c1n.south east)
			([xshift=-2mm,yshift=2mm]c21.north west) rectangle ([xshift=2mm,yshift=-5.5mm]c26.south east)
		;
		\draw[rounded corners, fill=gray!40] 
			(c11.north west) rectangle (c1n.south east)
			(c21.north west) rectangle (c22.south east)
			(c23.north west) rectangle (c241.south east)
			(c25.north west) rectangle (c26.south east)
		;
		\end{pgfonlayer}

		\node[above=of c11] (h11) {\phantom{\scalebox{1.5}{$ \bullet $}}};
		\node[above=of dotsC1] (h12) {\phantom{\scalebox{1.5}{$ \bullet $}}};
		\node[above=of c1n-1] (h13) {\phantom{\scalebox{1.5}{$ \bullet $}}};
		\node[above=of c1n] (h14) {\phantom{\scalebox{1.5}{$ \bullet $}}};
		
		\node[above=of c21] (h211) {\phantom{\scalebox{1.5}{$ \bullet $}}};
		\node[above=of dotsC21] (h212) {\phantom{\scalebox{1.5}{$ \bullet $}}};
		\node[above=of c22] (h213) {\phantom{\scalebox{1.5}{$ \bullet $}}};
		
		\node[above=of c23] (h221) {\phantom{\scalebox{1.5}{$ \bullet $}}};
		\node[above=of c231] (h222) {\phantom{\scalebox{1.5}{$ \bullet $}}};
		\node[above=of dotsC22] (h223) {\phantom{\scalebox{1.5}{$ \bullet $}}};
		\node[above=of c24] (h224) {\phantom{\scalebox{1.5}{$ \bullet $}}};
		\node[above=of c241] (h225) {\phantom{\scalebox{1.5}{$ \bullet $}}};
		
		\node[above=of c25] (h231) {\phantom{\scalebox{1.5}{$ \bullet $}}};
		\node[above=of c26] (h232) {\phantom{\scalebox{1.5}{$ \bullet $}}};
		
		\node[xshift=-0.5*\xdis] at (h11) (s1) {};
		\node[xshift=0.5*\xdis] at (h12) (s2) {};
		\node[xshift=0.5*\xdis] at (h14) (s3) {};
		
		\node[xshift=-0.5*\xdis] at (h211) (s4) {};
		\node[xshift=-0.5*\xdis] at (h212) (s5) {};
		\node[xshift=0.5*\xdis] at (h213) (s6) {};
		
		\node[xshift=-0.5*\xdis] at (h221) (s7) {};
		\node[xshift=-0.5*\xdis] at (h222) (s8) {};
		\node[xshift=0.5*\xdis] at (h223) (s9) {};
		\node[xshift=0.5*\xdis] at (h225) (s10) {};
		
		\node[xshift=-0.5*\xdis] at (h231) (s11) {};
		\node[xshift=0.5*\xdis] at (h232) (s12) {};
		
		\begin{pgfonlayer}{bg}
		\draw[rounded corners, fill=gray!60] 
		(h11.north west) rectangle (h12.south east)
		(h13.north west) rectangle (h14.south east)
		(h211.north west) rectangle (h211.south east)
		(h212.north west) rectangle (h213.south east)
		(h221.north west) rectangle (h221.south east)
		(h222.north west) rectangle (h223.south east)
		(h224.north west) rectangle (h225.south east)
		(h231.north west) rectangle (h232.south east)		
		;
		
		\end{pgfonlayer}

		\node[left=of h11,xshift=-\xdis] (dyncl) {\Large $ \dynsubclasses $};
		\node[right=of c26,xshift=\xdis] (baseclDots) {\huge $ \cdots $};
		\node[right=of h232,xshift=\xdis] (dynclDots) {\huge $ \cdots $};
				
		\node at ($ (h11)!0.5!(h12) $) (Z11) {$ \dynsubcl_1^1 $};
		\node at ($ (h13)!0.5!(h14) $) (Z12) {$ \dynsubcl_1^2 $};
		\node at (h211) (Z21) {$ \dynsubcl_2^1 $};
		\node at ($ (h212)!0.5!(h213) $) (Z22) {$ \dynsubcl_2^2 $};
		\node at (h221) (Z23) {$ \dynsubcl_2^3 $};
		\node at ($ (h222)!0.5!(h223) $) (Z24) {$ \dynsubcl_2^4 $};
		\node at ($ (h224)!0.5!(h225) $) (Z25) {$ \dynsubcl_2^5 $};
		\node at ($ (h231)!0.5!(h232) $) (Z26) {$ \dynsubcl_2^6 $};
		
		\path[segment aspect=0,->,snake=coil,segment amplitude=1.5pt,segment length=9pt,every path/.style={snake=coil}]
		($ (h11.south)!0.5!(h12.south) $) edge ($ (c11.north)!1/4!(c1n.north) $)
		($ (h13.south)!0.5!(h14.south) $) edge ($ (c11.north)!7/8!(c1n.north) $)
		(h211.south) edge (c21.north)
		($ (h212.south)!0.5!(h213.south) $) edge ($ (c21.north)!3/4!(c22.north) $)
		(h221.south) edge (c23.north)
		($ (h222.south)!0.5!(h223.south) $) edge ($ (c23.north)!3/8!(c241.north) $)
		($ (h224.south)!0.5!(h225.south) $) edge ($ (c23.north)!7/8!(c241.north) $)
		($ (h231.south)!0.5!(h232.south) $) edge ($ (c25.north)!1/2!(c26.north) $)
		;
		
	\end{tikzpicture}
	\caption{A schematic depiction of dynamic subclasses $ \dynsubclasses $ (top level) and the function $ \assdynsubcl $ (illustrated by wave arrows) in a dynamic representation
for a generic example.
	In the bottom level, basic color classes $ \basecl_{i} $ of a symmetric net, divided into static subclasses $ \basecl_{i,q} $, are depicted. 
	The dots $ \bullet $ are arbitrary colors belonging to the respective static subclass.
	For every dynamic subclass $ \dynsubcl_i^j $, the value $ q=\assdynsubcl(\dynsubcl_i^j) $ is indicated by a wave arrow to
	$ \basecl_{i,q} $, meaning that it represents $ |\dynsubcl_i^j| $ variables with values in $ \basecl_{i,q} $.
	The cardinality of each static subclass is the sum of the cardinalities of the corresponding dynamic subclasses:
	$ \sum_{\assdynsubcl(\dynsubcl_{i}^j)=q} |\dynsubcl_{i}^j|=|\basecl_{i,q}| $ for all $ i $.
	\label{fig-schematic-depiction}}
\end{figure}

\begin{example}\label{ex:Representation}
	Consider the symbolic decision set from the last example.
	We can naively build a dynamic representation by  taking one of the decision sets, and replacing each color by a dynamic subclass of cardinality $ 1 $.
	This results in as many dynamic subclasses as there are colors, i.e.,
	$ \dynsubclasses=\{ \dynsubcl_1^1,\dynsubcl_1^2,\dynsubcl_1^3,\dynsubcl_2^1 \} $
	with $ |\dynsubcl_i^j|=1 $ for all $ i,j $. 
	Since we have no partition into static subclasses, 
	i.e., $ \basecl_{1}=\basecl_{1,1} $ and $ \basecl_{2}=\basecl_{2,1} $,
	the function $ \assdynsubcl $ is trivially defined by $ \forall i,j: \assdynsubcl(\dynsubcl_i^j)=1 $.
	Below, the resulting dynamic decision set $ \decsetfct $ is depicted,
	with valid assignments that lead to elements $ \decset$, $\decset' $, and~$ \decset'' $.
	\begin{center}
		\begin{tikzpicture}[scale = 0.75, transform shape]\renewcommand{\xdis}{5cm}
		\tikzstyle{decisionset} = [rectangle,fill=lightgray,
		draw,align=center,minimum size=7mm]
		\tikzstyle{env} = [fill=white]
		\tikzstyle{sys} = [rounded corners]
		\node[decisionset,env,very thick,inner sep=2mm](Rep)	{
			$ \left( \mathit{R}.\dynsubcl_1^1, \{ g.\dynsubcl_1^1\} \right) $\\
			$ \left( \mathit{S}.\dynsubcl_1^1, \top \right)  $\\
			$ \left( \mathit{S}.\dynsubcl_1^2, \top \right)  $\\
			$ \left( \mathit{S}.\dynsubcl_1^3, \top \right)  $
		};
		\node[decisionset,env] at (Rep) [xshift=4.5cm] (D4)	{
			$ \left( \mathit{R}.c_1, \{ g.c_1\} \right) $\\
			$ \left( \mathit{S}.c_1, \top \right)  $\\
			$ \left( \mathit{S}.c_2, \top \right)  $\\
			$ \left( \mathit{S}.c_3, \top \right)  $
		};
		\node[decisionset,env] at (Rep) [xshift=10.9cm] (D4')	{
			$ \left( \mathit{R}.c_2, \{ g.c_2\} \right) $\\
			$ \left( \mathit{S}.c_1, \top \right)  $\\
			$ \left( \mathit{S}.c_2, \top \right)  $\\
			$ \left( \mathit{S}.c_3, \top \right)  $
		};
		\node[decisionset,env] at (Rep) [xshift=17.0cm] (D4'')	{
			$ \left( \mathit{R}.c_3, \{ g.c_3\} \right) $\\
			$ \left( \mathit{S}.c_1, \top \right)  $\\
			$ \left( \mathit{S}.c_2, \top \right)  $\\
			$ \left( \mathit{S}.c_3, \top \right)  $
		};
		\draw[->,thick]
		(Rep) edge node[above] {$ \dynsubcl_1^j \mapsfrom c_j $} (D4);
		\draw[->,thick]
		([yshift=-0.2mm]Rep.north east) -- node[below, pos=0.82, align=center] {$ \dynsubcl_1^2 \mapsfrom c_1,\dynsubcl_1^1\mapsfrom c_2,$ \\ $\dynsubcl_1^3\mapsfrom c_3 $} ([yshift=-0.2mm,xshift=7.5cm]Rep.north east) -- (D4');
		\draw[->,thick]
		([yshift=0.2mm]Rep.south east) -- node[above, pos=0.92, align=center] {$ \dynsubcl_1^3 \mapsfrom c_1,\dynsubcl_1^2\mapsfrom c_2,$ \\ $\dynsubcl_1^1\mapsfrom c_3 $} ([yshift=0.2mm,xshift=13.5cm]Rep.south east) -- (D4'')
		;
		\node at (Rep) [xshift = 1.75cm, yshift=-0.65cm](D4name) {\large$ \decsetfct $};
		\node at (D4) [xshift = 1.55cm, yshift=-0.65cm](D4name) {\large$ \decset $};
		\node at (D4') [xshift = 1.6cm, yshift=0.65cm](D4'name) {\large$ \decset' $};
		\node at (D4'') [xshift = -1.65cm, yshift=0.65cm](D4''name) {\large$ \decset'' $};
		\end{tikzpicture}
	\end{center}
	For example, the element $ (\mathit{R}.\dynsubcl_1^1, \{ g.\dynsubcl_1^1\}) $
	represents \emph{one} fixed arbitrary color $ c\in\basecl_{1}=\basecl_{1,1} $ (since $ |\dynsubcl_1^1|=1 $ and $ \assdynsubcl(\dynsubcl_1^1)=1 $) on place $ \mathit{R} $
	with $ g.c $ in its commitment set. The same color $ c $ resides on place $ \mathit{S} $, equipped with a $ \top $-symbol. 
\end{example}
As already mentioned, there are possibly many dynamic representations of the same symbolic decision set. 
In the next two subsections, we define minimal and ordered representations to the aim of determine a unique (canonical) one.
\subsubsection{Minimality}
We notice in the example above that $ \dynsubcl_1^2 $ and $ \dynsubcl_1^3 $
appear in the same context in $ \decsetfct $.
The \emph{context}~$ \decsetcon_\rep(\dynsubcl_{i}^j) $ 
of a dynamic subclass $ \dynsubcl_{i}^j $ is defined as the set of tuples
in $ \decsetfct $ where exactly one appearance of~$ \dynsubcl_{i}^j $
is replaced by a symbol $ \insertion $. In our example, $ \decsetcon_\rep(\dynsubcl_1^1) = \{
(\mathit{S}.\insertion, \top ),
(\mathit{R}.\insertion, \mathit{g}.c_1 ),
(\mathit{R}.c_1, \mathit{g}.\insertion )
\} $ and 
$ \decsetcon_\rep(\dynsubcl_1^2)=\{ (\mathit{S}.\insertion, \top ) \}=\decsetcon_\rep(\dynsubcl_1^3) $.
This means $ \dynsubcl_1^2 $ and $ \dynsubcl_1^3 $ can be \emph{merged}
into a new dynamic subclass of cardinality $ 2 $.
The resulting new dynamic representation is given by 
$ \dynsubclasses_{\mathit{min}}=\{ \dynsubcl_1^1,\dynsubcl_1^2,\dynsubcl_2^1 \} $
with $ |\dynsubcl_1^1|=|\dynsubcl_2^1|=1 $ and $ |\dynsubcl_1^2|=2 $,
and $ \decsetfct_{\mathit{min}}=\left\{ ( \mathit{R}.\dynsubcl_1^1, \{ g.\dynsubcl_1^1\} ),
( \mathit{S}.\dynsubcl_1^1, \top ),
( \mathit{S}.\dynsubcl_1^2, \top ) \right\} $.
As before, we have $ \assdynsubcl_{\mathit{min}}\equiv 1 $.
\begin{center}
	\begin{tikzpicture}[scale = 0.75, transform shape]\renewcommand{\xdis}{5cm}
	\tikzstyle{decisionset} = [rectangle,fill=lightgray,
	draw,align=center,minimum size=7mm]
	\tikzstyle{env} = [fill=white]
	\tikzstyle{sys} = [rounded corners]
	\node[decisionset,env,very thick](Rep)	{
		$ \left( \mathit{R}.\dynsubcl_1^1, \{ g.\dynsubcl_1^1\} \right) $\\
		$ \left( \mathit{S}.\dynsubcl_1^1, \top \right)  $\\
		$ \left( \mathit{S}.\dynsubcl_1^2, \top \right)  $
	};
	\node[decisionset,env] at (Rep) [xshift=5.5cm] (D4)	{
		$ \left( \mathit{R}.c_1, \{ g.c_1\} \right) $\\
		$ \left( \mathit{S}.c_1, \top \right)  $\\
		$ \left( \mathit{S}.c_2, \top \right)  $\\
		$ \left( \mathit{S}.c_3, \top \right)  $
	};
	\node[decisionset,env] at (Rep) [xshift=10.8cm] (D4')	{
		$ \left( \mathit{R}.c_2, \{ g.c_2\} \right) $\\
		$ \left( \mathit{S}.c_1, \top \right)  $\\
		$ \left( \mathit{S}.c_2, \top \right)  $\\
		$ \left( \mathit{S}.c_3, \top \right)  $
	};
	\node[decisionset,env] at (Rep) [xshift=16.5cm] (D4'')	{
		$ \left( \mathit{R}.c_3, \{ g.c_3\} \right) $\\
		$ \left( \mathit{S}.c_1, \top \right)  $\\
		$ \left( \mathit{S}.c_2, \top \right)  $\\
		$ \left( \mathit{S}.c_3, \top \right)  $
	};
	\draw[->,thick]
	(Rep) edge node[align=center] {$ \dynsubcl_1^j \mapsfrom c_1 $\\$ \dynsubcl_1^2 \mapsfrom c_2,c_3 $} (D4);
	\draw[->,thick]
	(Rep.north) -- ([yshift=4mm]Rep.north) -- node[below, pos=0.90, align=center] {$\dynsubcl_1^1\mapsfrom c_2$ \\$ \dynsubcl_1^2 \mapsfrom c_1,c_3,$  } ([xshift=7.5cm,yshift=4mm]Rep.north east) -- (D4');
	\draw[->,thick]
	(Rep.south) -- ([yshift=-4mm]Rep.south) -- node[above, pos=0.93, align=center] { $\dynsubcl_1^1\mapsfrom c_3 $\\$ \dynsubcl_1^2 \mapsfrom c_1, c_2,$ } ([xshift=13cm,yshift=-4mm]Rep.south east) -- (D4'')
	;
	\node at (Rep) [xshift = 1.9cm, yshift=-0.9cm](D4name) {\large$ \decsetfct_{\mathit{min}} $};
	\node at (D4) [xshift = 1.45cm, yshift=-0.65cm](D4name) {\large$ \decset $};
	\node at (D4') [xshift = 1.5cm, yshift=0.65cm](D4'name) {\large$ \decset' $};
	\node at (D4'') [xshift = -1.55cm, yshift=0.65cm](D4''name) {\large$ \decset'' $};
	\end{tikzpicture}
\end{center}

\emph{Minimal representations} $ \rep $ do not contain 
any two dynamic subclasses $ \dynsubcl_{i}^j,\dynsubcl_{i}^k $ satisfying both
$ \assdynsubcl(\dynsubcl_i^j)= \assdynsubcl(\dynsubcl_i^j)
$ and $ \decsetcon_\rep(\dynsubcl_i^j)= \decsetcon_\rep(\dynsubcl_i^k) $.
Such two dynamic subclasses represent two sets of colors in the same static subclass
that appear in exactly the same context, and can therefore be merged.
For ordered classes $ \basecl_i $ we additionally require $ k=\succfunc(j) $.
Given a dynamic representation, it is algorithmically simple to construct a 
minimal representation of the same symbolic decision set 
by iteratively merging all such pairs of dynamic subclasses.
The dynamic representation above that resulted from merging the subclasses
is therefore minimal.
Still, minimality is not enough to obtain a unique canonical representation,
since we can permute the indices $ j $ of the dynamic subclasses~$ \dynsubcl_{i}^j $.

A \emph{permutation} of the dynamic subclasses is 
a tuple $ \rho=(\rho_1,\dots,\rho_n) $
of permutations on $ \{ \dynsubcl_{i}^j\with 1\leq j\leq m_i \} $
that respects the function $ \assdynsubcl $,
i.e., $ \assdynsubcl(\dynsubcl_i^j)=\assdynsubcl(\rho_i(\dynsubcl_i^j)) $.
In the case of an ordered class $ \basecl_{i} $ we again only consider rotations $ \rho_i $.
We can apply a permutation to a representation $ \rep=(\dynsubclasses,\assdynsubcl,\decsetfct) $
by keeping~$ \dynsubclasses $ and replacing every occurrence of $ \dynsubcl_i^j $
in $ \decsetfct $ by $ \rho_i(\dynsubcl_i^j) $,
and changing the cardinality $ |\dynsubcl_i^j| $ accordingly.
The function $ \assdynsubcl $ does not change.
We denote the new representation by~$ \rho(\rep) $.

We now prove that minimal representations are unique up to a permutation.
This result is obtained using the observation that
every minimal representation can be reached 
from a dynamic representation that contains only 
dynamic subclasses of cardinality $ 1 $ (as the one in \refEx{Representation}) by merging.
\begin{lemma}\label{lem:minimalAlmostUnique}
	The minimal representations of a symbolic decision set are unique up to a permutation of the dynamic subclasses.
\end{lemma}
\begin{proof}
	It follows directly from the definition that 
	if $ \rep $ is minimal, so is $ \rho(\rep) $ for any permutation $ \rho $.
	Let now $ \rep $ and $ \rep' $ be two minimal representations
	of the same symbolic decision set. 
	We show there is a permutation between $ \rep $ and $ \rep' $.
	There exist two corresponding representations 
	$ \rep_s $ and $ \rep_s' $,
	where all dynamic subclasses are split into 
	dynamic subclasses of cardinality $ 1 $.
	This means that $ \rep $ can be reached from 
	$ \rep_s $ by merging all subclasses 
	$ \dynsubcl_i^{j_1},\dots,\dynsubcl_i^{j_N} $
	such that each pair 
	$ \dynsubcl_i^{j_k},\dynsubcl_i^{j_\ell} $
	satisfies 
	$ \assdynsubcl(\dynsubcl_i^{j_k})= \assdynsubcl(\dynsubcl_i^{j_\ell})
	\land \decsetcon_{\rep_s}(\dynsubcl_i^{j_k})= \decsetcon_{\rep_s}(\dynsubcl_i^{j_\ell}) $,
	and analogously for $ \rep' $ and $ \rep_s' $.
	All representations of a fixed symbolic decision set
	that only contain dynamic subclasses of cardinality~$ 1 $
	can trivially transformed into each other by a permutation.
	Furthermore, we can pick a permutation $ \rho $
	from~$ \rep_s $ to~$ \rep_s' $ such that 
	if $ \dynsubcl_i^{j_k} $ and $ \dynsubcl_i^{j_\ell} $
	are merged in $ \rep_s $, then $ \rho(\dynsubcl_i^{j_k}) $
	and $ \rho(\dynsubcl_i^{j_\ell}) $ are merged in $ \rep_s' $.
	This implies a corresponding permutation between 
	$ \rep $ and $ \rep' $.
\end{proof}

\subsubsection{Ordering}
We can choose one of the minimal representations by \emph{ordering}
the dynamic subclasses. 
In the following, we present one possible way to do so.
We display the dynamic decision set $ \decsetfct $ as a matrix,
with rows and columns indicating (tuples of) dynamic subclasses $ \dynsubcl $.
An element of the matrix at entry~$ (Z,Z') $ returns all tuples $ (p,t) $
satisfying $ (p.Z,\comset)\in\decsetfct $  and $ t.\dynsubcl'\in\comset $
for a commitment set $ \comset $.
Also, all tuples $ (p,\top) $ satisfying $ {(p.\dynsubcl,\top)\in\decsetfct} $
and all tuples $ (p,\emptyset) $ satisfying $ (p.\dynsubcl,\emptyset)\in\decsetfct $
are returned (in these cases, $ \dynsubcl' $ is neglected). 

The elements of the matrix therefore are in $ \powerset( \places\!\times\!(\transitions\cup\{\top,\emptyset\})) $.
Since this set is finite, we can give an arbitrary, but fixed, total order $ < $ on it.
This order can be extended to the matrices over the set
(the lexicographic order by row-wise traversion through a matrix).
Then we can determine a permutation
such that, when applied to the dynamic subclasses,
the matrix is \emph{minimal with respect to the lexicographic order}.
The corresponding dynamic representation is called \emph{ordered}.
\begin{example}
	On the left we see the matrix of the dynamic decision set $ \decsetfct $
	in the minimal representation given above.
	The first entry, at $ (\dynsubcl_1^1, \dynsubcl_1^1) $, e.g., is 
	$ \{ (\mathit{S},\top), (\mathit{R},\mathit{g}) \} $ since
	$ (\mathit{S}.\dynsubcl_1^1,\top) $ and $ (\mathit{R}.\dynsubcl_1^1,\mathit{g}.\dynsubcl_1^1) $
	are in $ \decsetfct $.
	Assume $ \{(\mathit{S},\top)\}<\{ (S,\top),(\mathit{R},g) \} $.
	When the permutation switching $ \dynsubcl_1^1 $ and $ \dynsubcl_1^2 $
	is applied, we get the right matrix, which is lexicographically smaller
	(the first entry is smaller).
	\begin{center}
		\begin{tabular}{c|cc}
			& $ \dynsubcl_1^1 $                                    & $ \dynsubcl_1^2 $           \\ \hline
			$ \dynsubcl_1^1 $ & $ \{ (\mathit{S},\top), (\mathit{R},\mathit{g}) \} $ & $ \{ (\mathit{S},\top) \} $ \\
			$ \dynsubcl_1^2 $ & $ \{ (\mathit{S},\top) \} $                          & $ \{ (\mathit{S},\top) \} $
		\end{tabular}
		$ \overset{\dynsubcl_1^1\leftrightarrow \dynsubcl_1^2}{\quad\scalebox{1.5}{\begin{tikzpicture}[]
			\draw[<->,>=To]
			(0,0)--(1,0)
			;
			\end{tikzpicture}}\quad} $
		\begin{tabular}{c|cc}
			& $ \dynsubcl_1^1 $                                    & $ \dynsubcl_1^2 $           \\ \hline
			$ \dynsubcl_1^1 $ & $ \{ (\mathit{S},\top) \} $ & $ \{ (\mathit{S},\top) \} $ \\
			$ \dynsubcl_1^2 $ & $ \{ (\mathit{S},\top) \} $                          & $ \{ (\mathit{S},\top), (\mathit{R},\mathit{g}) \} $
		\end{tabular}
	\end{center}
	Thus, the minimal representation from above 
	is transformed into a \emph{minimal and ordered} representation 
	$ (\dynsubclasses_{\mathit{ord}},\assdynsubcl_{\mathit{ord}}\decsetfct_{\mathit{ord}}) $
	by the permutation $ \dynsubcl_1^1\leftrightarrow \dynsubcl_1^2 $.
\end{example}

We have defined what it means for a dynamic representation to be minimal and to be ordered.
To the aim of showing that this leads to canonical representations of symbolic decision sets,
we formulate the following lemma.

\begin{lemma}\label{lem:symmetricOrderedRepsSameDynDecset}
	Let $ \rep_1=(\dynsubclasses_1,\assdynsubcl_1,\decsetfct_1) $ and $ \rep_2=(\dynsubclasses_2,\assdynsubcl_2,\decsetfct_2) $ 
	be two ordered representations of the 
	same symbolic decision set, and 
	$ \rho $ be a permutation 
	such that $ \rep_2=\rho(\rep_1) $. Then
	$ \decsetfct_1=\decsetfct_2 $.
\end{lemma}

\begin{proof}
	We denote the matrix of a representation $ \rep_i $ by $ \mathit{mat}_i $.
	Since $ \rep_1 $ is ordered, 
	$ \rho $ cannot transform $ \mathit{mat}_1 $
	into a (lexicographically) smaller matrix. 
	Aiming a contradiction, assume that $ \rho $ transforms $ \mathit{mat}_1 $ into a bigger matrix $ \rho(\mathit{mat}_1) $.
	Since $ \rho(\mathit{mat}_1)=\mathit{mat}_2 $,
	this would mean that $ \rho^{-1} $ transforms $ \mathit{mat}_2 $
	into a smaller matrix, implying that $ \rep_2 $ is not ordered. Contradiction.
	Therefore, $ \mathit{mat}_1=\rho(\mathit{mat}_1)=\mathit{mat}_2 $.
	Since $ \mathit{mat}_i $ is just another presentation of $ \decsetfct_i $,
	we follow $ \decsetfct_1=\decsetfct_2 $.
\end{proof}

Now we are equipped with all tools to prove
that there is exactly one minimal and ordered representation
for each symbolic decision set.
	We call this representation the \emph{canonical} (dynamic) representation.
\begin{theorem}\label{theo:CanonicalRepresentationUnique}
	For every symbolic decision set there is 
	exactly one minimal and ordered dynamic representation. 
\end{theorem}
\begin{proof}
	Let $ \rep_1=(\dynsubclasses_1,\decsetfct_1) $ 
	and $ \rep_2=(\dynsubclasses_2,\decsetfct_2) $
	be two minimal and ordered representations 
	of the same symbolic decision set.
	\refLemma{minimalAlmostUnique} gives us that there
	is a permutation~$ \rho $
	such that $ \rep_2=\rho(\rep_1) $.
	Then we have by
	\refLemma{symmetricOrderedRepsSameDynDecset}
	that $ \decsetfct_1=\decsetfct_2 $.
	This means for all 
	$ \dynsubcl_i^j $ that 
	$ \decsetcon_{\rep_1}(\dynsubcl_i^j)
	=\decsetcon_{\rep_2}(\rho_i(\dynsubcl_i^j)) $.
	$ \rho $ is a permutation, so by definition
	$ \assdynsubcl(\rho_i(\dynsubcl_i^j))
	= \assdynsubcl(\dynsubcl_i^j) $.
	$ \rep_1 $ is minimal, 
	therefore it must hold that $ \rho_i(\dynsubcl_i^j)=\dynsubcl_i^j $,
	else we could merge the subclasses 
	$ \rho_i(\dynsubcl_i^j)$ and~$\dynsubcl_i^j $.
	This holds for all $ i,j $,
	such that we finally have 
	$ \forall i:\rho_i=\identity_{\{ \dynsubcl_{i}^j\with 1\leq j\leq m_i \}} $
	and therefore, $ \rep_1=\rho(\rep_1)=\rep_2 $.
\end{proof}

We can algorithmically order a minimal representation by calculating all symmetric representations
and by finding the one with the lexicographically smallest matrix.
These are maximally $ |\symmetries_\SN| $ comparisons of dynamic representations.
So by first making a dynamic representation minimal, and then ordering it, we get the respective canonical representation.

\begin{corollary}\label{cor:CostOfCanon}
For a given symbolic decision set, we can construct the canonical representation in~$ O(|\symmetries_\SN|) $.
\end{corollary}
\subsection{Relations between Canonical Representations}
In the previous section, we introduced canonical representations of equivalence classes of decision sets.
Decision sets describe situations in a Petri game and are employed as 
nodes in the two-player game used to solve said Petri game.
The edges in the game are build from relations between the decision sets,
which are demonstrated in \refEx{DecisionFirings}.
Since the goal is to replace (symbolic) decision sets by the canonical representations in the two-player game,
we now define respective relations on this level.

Between decision sets, we have the two relations
$ \decset[t.v\rangle\decset' $
(with $ v\in\valuations(t)=\basecl_1^{t_1}\!\times\!\cdots\!\times\!\basecl_n^{t_n} $) and $ \decset[\top\rangle\decset' $.
In canonical representations, we abstract from specific colors $ c\in\basecl_i $ and replace
them by dynamic subclasses $ \dynsubcl_i^j $ of variables.
However, in the process of $ \top $-resolution or transition firing,
two objects represented by the same dynamic subclass can act differently.
This means we \emph{instantiate} special objects in the classes
that are relevant in the $ \top $-resolution or transition firing.

For this, each dynamic subclass $ \dynsubcl_{i}^j $ in a canonical representation $ \rep $ is \emph{split}
into finitely many \emph{instantiated} $ \dynsubcl_{i}^{j,k} $ of cardinality $|\dynsubcl_{i}^{j,k}|= 1 $ with $ k>0 $,
and a subclass~$ \dynsubcl_{i}^{j,0} $, containing the possibly remaining, non-instantiated, variables.
Then, a $ \top $ is resolved, or a transition is fired, 
with the dynamic subclasses $ \dynsubcl_{i}^{j,k} $ replacing
explicit data entries $ c\in\basecl_{i} $.
Finally, the canonical representation $ \rep' $ of the reached dynamic representation is found.
These relations are denoted by $ \rep[\top\rangle\rep' $
and $ \rep[t.\dynsubcl^{[\instancesubcl,\instanceelement]}\rangle\rep' $,
where $ \dynsubcl^{[\instancesubcl,\instanceelement]} $ is a tuple of instantiated $ \dynsubcl_i^{j,k} $.

Below we see an example that corresponds to the last two steps in \refEx{DecisionFirings}.
We calculated the second canonical representation in the last section.
It is reached from the first canonical representation by firing $ \mathit{inf}\!.\dynsubcl_1^{2,1} $.
In this process one (the only) element in $ \dynsubcl_1^2 $ is instantiated by a dynamic subclass~$ \dynsubcl_1^{2,1} $
of cardinality $ 1 $. After the actual firing, the reached representation is made canonical.
Then, a $ \top $ is resolved. Here, $ \dynsubcl_1^1 $
is split into $ \dynsubcl_1^{1,1} $ and $ \dynsubcl_1^{1,2} $ with
$ | \dynsubcl_1^{1,1} |=| \dynsubcl_1^{1,2} | =1$. 
In the reached dynamic representation, no two subclasses have the same context, 
so it is already minimal. After ordering we get the canonical representation.
\begin{center}
	\begin{tikzpicture}[scale = 0.75, transform shape]\renewcommand{\xdis}{5cm}
	\tikzstyle{decisionset} = [rectangle,fill=lightgray,
	draw,align=center,minimum size=7mm]
	\tikzstyle{env} = [fill=white]
	\tikzstyle{sys} = [rounded corners]
	\node[decisionset,env] (RepPre)	{
		$ \left( \mathit{I}.\dynsubcl_1^2, \{ \mathit{inf}\!.\dynsubcl_1^2\} \right) $\\
		$ \left( \mathit{S}.\dynsubcl_1^1, \{ \mathit{inf}\!.\dynsubcl_1^1,\mathit{inf}\!.\dynsubcl_1^2 \} \right)  $\\
		$ \left( \mathit{S}.\dynsubcl_1^2, \{ \mathit{inf}\!.\dynsubcl_1^1,\mathit{inf}\!.\dynsubcl_1^2 \} \right)  $
	};
	\node at (RepPre) (dynsubRepPre) [align=right,xshift=-3cm,yshift=0mm]{$ |\dynsubcl_1^1|=2 $\\$ |\dynsubcl_1^2|=1 $};
	\node[decisionset,env] at (RepPre) [xshift=7cm]	(Rep)	{
		$ \left( \mathit{R}.\dynsubcl_1^1, \{ g.\dynsubcl_1^2\} \right) $\\
		$ \left( \mathit{S}.\dynsubcl_1^1, \top \right)  $\\
		$ \left( \mathit{S}.\dynsubcl_1^2, \top \right)  $
	};
	\node at (Rep) (dynsubRep) [align=right,xshift=-2.25cm,yshift=-4mm]{$ |\dynsubcl_1^1|=2 $\\$ |\dynsubcl_1^2|=1 $};
	\node[decisionset,env] at (Rep) [xshift=7cm]	(RepPost)	{
		$ \left( \mathit{R}.\dynsubcl_1^3, \{ g.\dynsubcl_1^3\} \right) $\\
		$ \left( \mathit{S}.\dynsubcl_1^1, \{a.(\dynsubcl_1^1,\dynsubcl_1^3)\} \right)  $\\
		$ \left( \mathit{S}.\dynsubcl_1^2, \{a.(\dynsubcl_1^2,\dynsubcl_1^3)\} \right)  $\\
		$ \left( \mathit{S}.\dynsubcl_1^3, \{a.(\dynsubcl_1^3,\dynsubcl_1^3)\} \right)  $
	};
	\node at (RepPost) (dynsubRep) [align=right,xshift=-3.2cm,yshift=-5mm]{$ \forall j\in\{1,2,3\}: $\\$ |\dynsubcl_1^j|=2  $};
	\path[->]
	([yshift=3mm]RepPre.east) edge node[above,yshift=-0.5mm]{$ \mathit{inf}.(\dynsubcl_1^{2,1}) $} ([yshift=3mm]Rep.west)
	([yshift=3mm]Rep.east) edge node[above]{$ \top $} ([yshift=3mm]RepPost.west)
	;
	\end{tikzpicture}
\end{center}

We now go into detail on the definitions of the 
symbolic versions of these symbolic relations.
In the \emph{symbolic transition firing},
we consider \emph{symbolic instances}
of a transition~$ t $, where
the parameters are assigned 
elements in
dynamic subclasses
instead of specific colors
in basic color classes.
For a representation $ \rep $
we group the dynamic subclasses by
$ \forall 1\leq i\leq n: 
\dynsubclasses_i=\{ \dynsubcl_i^j\with 1\leq j\leq m_i \} $.
For a transition $ t\in\hltransitions $
we then define $ \valuations_\rep(t)=\dynsubclasses_1^{t_1}\!\times\!\cdots\!\times\!\dynsubclasses_n^{t_n} $
(analogously to $ \valuations(t)=\basecl_1^{t_1}\!\times\!\cdots\!\times\!\basecl_n^{t_n} $).
The instance of the 
$ x $-th parameter 
of type $ \dynsubclasses_i $ 
in $ \valuations_\rep(t) $ (i.e., $ 1\leq x\leq t_i $)
can be specified by a pair 
$ (\instancesubcl_i(x),
\instanceelement_i(x))=(j,k) $,
meaning that the parameter represents
the \mbox{$ k $-th} (arbitrarily chosen)
element of~$ \dynsubcl_i^j $.
Since we maximally have 
$ |\dynsubcl_i^j| $
different instances of $ \dynsubcl_i^j $,
$ k < | \dynsubcl_i^j | $ must hold.
Furthermore, $ k $ must be
less than the number of parameters
instanced to $ \dynsubcl_i^j $.
A \emph{symbolic mode}
$ [\instancesubcl,\instanceelement] $ of~$ t $ 
is defined by two families of functions 
$\instancesubcl=
\{ \instancesubcl_i:
\{ 1,\dots,t_i \}\to\N^+
\} $
and 
$ \instanceelement=
\{ \instanceelement_i:
\{ 1,\dots,t_i \}\to\N^+
\} $
satisfying the conditions above.
We denote $ \dynsubcl^{[\instancesubcl,\instanceelement]}=
((\dynsubcl_i^{\instancesubcl_i(x),
	\instanceelement_i(x)})_{x=1}^{t_i})_{i=1}^n $,
and $ t.\dynsubcl^{[\instancesubcl,\instanceelement]} $
is called a \emph{symbolic instance of}~$ t $.

To fire a symbolic instance~$ 
t.\dynsubcl^{[\instancesubcl,\instanceelement]} $
from a representation~$ \rep=(\dynsubclasses,\assdynsubcl,\decsetfct) $,
we have to \emph{split} the dynamic subclasses~$ \dynsubcl_i^j $
such that the 
(by $ \instancesubcl_i $ and $ \instanceelement_i $) 
instantiated elements are represented by
new subclasses~$ \dynsubcl_i^{j,k} $ of 
cardinality~$ 1 $.
The possibly remaining (non-instantiated) elements
are collected in the additional subclass~$ 
\dynsubcl_i^{j,0} $.
The new dynamic subclasses $ \dynsubcl_i^{j,k} $ are assigned to the
same static subclasses as $ \dynsubcl_i^j $ by $ \assdynsubcl $.
The new dynamic decision set
can be naturally derived from $
\decsetfct $: since the subclasses~$ \dynsubcl_i^j $ 
are split into $ \{\dynsubcl_i^{j,k}\}_k $, 
every tuple~$ (p.\dynsubcl,\comset)\in\decsetfct $ 
in which the original subclasses appeared
is replaced by all possible tuples containing the new 
subclasses instead. 
In the case of an ordered class $ \basecl_i $,
every~$ \dynsubcl_i^j $ is split into $ |\dynsubcl_i^j| $ dynamic subclasses
of cardinality $ 1. $
This \emph{split representation} is denoted by~$ \rep_{[ \instancesubcl,\instanceelement ]} $.

A symbolic instance $ t.\dynsubcl^{[\instancesubcl,\instanceelement]} $
can be \emph{fired} from the representation
$ \rep_{[ \instancesubcl,\instanceelement ]}=:\rep_s=(\dynsubclasses_s,\assdynsubcl_s,\decsetfct_s) $
analogously to the ordinary case $ \decset[t.c\rangle $ 
with the dynamic subclasses $ \dynsubcl_{i}^{j,k} $, $ k>0 $ (i.e., $ |\dynsubcl_{i}^{j,k}|=1 $) with $ \assdynsubcl( \dynsubcl_{i}^{j,k} )=q $ replacing 
the explicit colors $ c\in\basecl_{i,q} $ in guards and arc expressions.
This procedure is sound since we consider symmetric nets, 
where guards and arc expressions handle all colors in the same static subclass equally.
We arrive at a new representation $ \rep_s' $ with same 
dynamic subclasses~$ \dynsubclasses_s'=\dynsubclasses_s $
and $ \assdynsubcl_s'=\assdynsubcl_s $,
and the dynamic decision set $ \decsetfct_s' $ that results from firing $ t.\dynsubcl^{[\instancesubcl,\instanceelement]} $
from~$ \decsetfct_s $ (w.r.t.\ the flow function $ \hlflowfunc $).
Finally, the canonical representation $ \rep' $ of $ \rep_s' $ is found.
The symbolic firing is denoted by $ \rep[t.\dynsubcl^{[\instancesubcl,\instanceelement]}\rangle\rep' $ and
can be described by
\begin{equation*}
\rep\xrightarrow{\text{splitting w.r.t } [\instancesubcl,\instanceelement]} 
\rep_s \xrightarrow{\text{firing } t.\dynsubcl^{[\instancesubcl,\instanceelement]}}
\rep_s'\xrightarrow{\text{canon. rep.}}\rep'.
\end{equation*}
The symbolic instances of a transition $ t $
form a partition of the regular instances $ t.c $, $ c\in\valuations(t) $.

The idea of \emph{symbolic $ \top $-resolution}
is similar to the
symbolic firing.
A canonical representation which 
contains a tuple $ (p.\dynsubcl,\top) $
is split into a finer 
representation.
The symbol $ \top $
gets resolved as before, 
and finally, the new canonical representation is built.
While in $ \rep[\instancesubcl,\instanceelement] $, only some
dynamic subclasses of cardinality $ 1 $
are split off and the rest is collected in a subclass $ \dynsubcl_i^{j,0} $, 
in the symbolic $ \top $-resolution
every dynamic subclass $ \dynsubcl_{i}^j $ in $ \rep $ is split into $ |\dynsubcl_{i}^j| $ dynamic subclasses of cardinality $ 1 $,
just as in the case of an ordered class for the symbolic firing.

This gives, for a canonical representation $ \rep $,
that $ \rep[\top\rangle $ iff
$ \exists (p.\dynsubcl,\top)\in\decsetfct $,
and $ \rep[\top\rangle\rep' $ 
iff from the splitted representation $ \rep_s $,
a representation $ \rep_s' $ can be reached by copying the dynamic subclasses and in $ \decsetfct_s' $
every $ (p.\dynsubcl,\top)\in\decsetfct_s $ chooses
a new commitment set $ \comset $,
and $ \rep' $ is 
the canonical representation of $ \rep_s' $.
The symbolic $ \top $-resolution $ \rep[\top\rangle\rep' $ is therefore described by
\begin{equation*}
\rep\xrightarrow{\text{splitting}} 
\rep_s \xrightarrow{\top\text{-resolution}}
\rep_s'\xrightarrow{\text{canon. rep.}}\rep'.
\end{equation*}  

The next theorem states that each relation
between two decision sets is represented by a relation between the 
corresponding canonical representations of the symbolic decision sets (and the other way round).
The proof for the case $ \decset [t.v\rangle \decset' $
follows by applying a valid assignment $ \widehat{\validassignment_{\decset}} $
to~$ v $ and splitting~$ \rep $ correspondingly. 
The case $ \decset [\top\rangle \decset' $ works analogously.

\begin{theorem}\label{theo:relBetweenCanonReps}
	Let	$ \decset $ and $ \decset' $ be decision sets 
	and let $ \rep $ and $ \rep' $ be the canonical representations of the respective symbolic decision set.
	\begin{enumerate}
		\item \begin{enumerate}[i)]
			\item $ \decset [t.v\rangle \decset' 
			\Rightarrow 
			\exists \instancesubcl,\instanceelement : \rep [t.\dynsubcl^{[\instancesubcl,\instanceelement]}\rangle \rep' \land \exists \validassignment_{\decset}: \validassignment_{\decset}(\decset)=\symbdecset_{[\instancesubcl,\instanceelement]} \land \validassignment_{\decset}^{-1}(\dynsubcl^{[\instancesubcl,\instanceelement]})=v $ 
			\item $  \rep [t.\dynsubcl^{[\instancesubcl,\instanceelement]}\rangle \rep'
			\Rightarrow 
			\exists \validassignment_{\decset}: \validassignment_{\decset}(\decset)=\symbdecset_{[\instancesubcl,\instanceelement]} 
			\land \exists s\in\symmetries_\SN: \decset[t.\validassignment_{\decset}^{-1}(\dynsubcl^{[\instancesubcl,\instanceelement]})\rangle s(\decset') $
		\end{enumerate}
		\item \begin{enumerate}[i)]
			\item $ \decset [\top \rangle \decset' \Rightarrow \rep [\top \rangle \rep' $
			\item $ \rep [\top \rangle \rep' \Rightarrow \exists s\in\symmetries_\SN: \decset [\top \rangle s(\decset') $
		\end{enumerate}
	\end{enumerate}
\end{theorem}
\begin{proof}
	Case 1.i): For a relation $ \decset [t.v\rangle \decset' $ consider the valid assignment 
	$ \widehat{\validassignment_{\decset}}: \bigcup_{i=1}^n \basecl_i\to \dynsubclasses $
	with $ \rep=(\dynsubclasses,\assdynsubcl,\decsetfct) $ and $ \widehat{\validassignment_{\decset}}(\decset)=\decsetfct $.
	Applying $ \widehat{\validassignment_{\decset}} $ to a valuation $ v $ gives a 
	tuple of dynamic subclasses in $ \dynsubclasses $.
	Splitting the contained elements in this tuple into dynamic subclasses~$ \dynsubcl_i^{j,k} $ of cardinality $ 1 $
	with respect to the different colors in $ v $, gives a tuple~$ \dynsubcl^{[\instancesubcl,\instanceelement]} $ of instantiated subclasses.
	This splitting is possible since~$ \validassignment_\decset $ is valid.
	This also implies $ \validassignment_{\decset}(\decset)=\symbdecset_{[\instancesubcl,\instanceelement]}  \land \validassignment_{\decset}^{-1}(\dynsubcl^{[\instancesubcl,\instanceelement]})=v$
	for the respectively ``split'' assignment $ \validassignment_{\decset} $.
	In the firing of $ t.\dynsubcl^{[\instancesubcl,\instanceelement]} $, the dynamic subclasses replace colors, which means that
	the reached dynamic representation represents the symbolic decision set of $ \decset' $.
	Making it canonical therefore results in $ \rep' $. 
	Case 1.ii) follows directly from the construction.
	The proof for the $ \top $-resolution works analogously.
\end{proof}

\subsection{Properties of Canonical Representations}\label{sec:RepProps}
The goal is to use canonical representations instead of 
individual decision sets or (arbitrary representatives of) symbolic decision sets as nodes in a two-player game.
The edges $ (\rep,\rep') $ in this game are built from relations 
$ \rep [t.\dynsubcl\rangle \rep' $ and $ \rep [\top\rangle \rep' $,
depending on the properties of $ \rep $.
For example, if $ \rep $ describes nondeterministic situations
in the Petri game, then the edges from $ \rep $ are built in such a way that
player~0 (representing the system) cannot win the game from there. 
In this section, we define the relevant properties of canonical representations.

For a decision set $ \decset $, let $ \marking(\decset)=\{ \place\in\places \with (\place,\top)\in\decset \lor \exists \comset\subseteq\transitions : (\place,\comset)\in\decset \} $ be the \emph{underlying marking}.
Decision sets can have the following properties.
\begin{definition}[Properties of Decision Sets \cite{DBLP:journals/iandc/FinkbeinerO17}]\label{def:propsOfDecsets} Let $ \decset $ be a decision set.
	\begin{itemize}
		\item \(\decset\) is \emph{en\-vi\-ron\-ment-de\-pen\-dent} iff
		$ \neg \decset[\top\rangle $, i.e.,
		there is no $ \top $ symbol in~$ \decset $,
		and $ \forall \transition\in\transitions : \decset[\transition\rangle \Rightarrow \preset{\transition}\cap\envplaces\neq\emptyset $,
		i.e., all enabled transitions have an environment place in their preset.
		\item \(\decset\) \emph{contains a bad place} iff $ \badplaces\cap\marking(\decset)\neq\emptyset $.
		\item \(\decset\)~is a \emph{deadlock} iff 
		$ \neg \decset[\top\rangle $,
		and $\exists \transition' \in \transitions : \marking(\decset)[\transition'\rangle \land  \forall \transition \in \transitions: \neg\decset[\transition\rangle $, i.e.,
		there is a transition that is enabled in the underlying marking, but the system forbids all enabled transitions.
		\item \(\decset\) is \emph{terminating} iff \(\forall \transition\in\transitions:\neg \marking(\decset)[\transition \rangle\).
		\item \(\decset\) is \emph{nondeterministic} iff 
		\(\exists \transition_1, \transition_2 : \transition_1\neq \transition_2
		\land \sysplaces\cap\preset{\transition_1}\cap\preset{\transition_2}\neq \emptyset \land 
		\decset[\transition_1\rangle \land \decset[\transition_2\rangle  \),
		i.e., two separate transitions sharing a system place in their presets
		both are~enabled.
		\end{itemize}
\end{definition}

In \cite{DBLP:journals/acta/GiesekingOW20}, 
we showed that all
decision sets in one equivalence class
share the same of the properties defined above.
Thus, we say a symbolic decision set has one of the above properties
iff its individual members have the respective property.

We now define these properties for canonical representations.
Since we do not want to consider individual decision sets, 
we do that on the level of dynamic representations.
The properties ``en\-vi\-ron\-ment-de\-pen\-dent'' and ``containing a bad place'' are rather analogous to the respective property of decision sets.
For termination and deadlocks we introduce 
for a given $ \rep $ 
the representation $ \rep_\mathit{all} $
with the same dynamic subclasses and assignment, and the dynamic decision set 
where every player has all possible transitions $ t.\dynsubcl $
in their commitment set.
Since then all transitions that could fire in the underlying marking are enabled,
$ \rep_\mathit{all} $ substitutes for $ \marking(\decset) $ in Def.~\ref{def:propsOfDecsets}. 
For nondeterminism we have to consider two cases. 
The first one ($ \mathit{ndet}_1(\rep) $) is analogous to the property for individual decision sets.
The second case ($ \mathit{ndet}_2(\rep) $) considers the situation that \emph{two instances} 
of one $ t.\dynsubcl $ can both fire with a shared system place in their preset.

\begin{definition}[Properties of Canonical Representations of Symbolic Decision Sets]
	Let $ \rep=(\dynsubclasses,\assdynsubcl,\decsetfct) $ be a canonical representation.
	\begin{itemize}
		\item 	$ \rep $ is \emph{en\-vi\-ron\-ment-de\-pen\-dent} iff
		$ \neg\rep[\top\rangle $, i.e.,
		there is no $ \top $ symbol in any tuple in $ \decsetfct $,
		and $ \forall t.\dynsubcl^{[\instancesubcl,\instanceelement]} : \rep[t.\dynsubcl^{[\instancesubcl,\instanceelement]}\rangle \Rightarrow 
		\exists p\in\hlenvplaces: (p,t)\in\hlflowfunc $.
		\item $ \rep $ \emph{contains a bad place} iff 
		$ \exists p\in\hlbadplaces\,\exists X: 
		(p.X,\top)\in\rep \lor\exists \comset : (p.X,\comset)\in\decsetfct  $.
		\item $ \rep $ is a \emph{deadlock} iff 
		$ \neg\rep[\top\rangle $,
		and $\exists t'.\dynsubcl' : \rep_\mathit{all}[t'.\dynsubcl'\rangle \land  \forall t.\dynsubcl : \lnot\rep[t.\dynsubcl\rangle $.
		\item $ \rep $ is \emph{terminating} 
		iff \(\forall t.\dynsubcl :\neg \rep_\mathit{all}[t.\dynsubcl \rangle\).
		\item $ \rep $ is \emph{nondeterministic} iff $ \mathit{ndet}_1(\rep)\lor \mathit{ndet}_2(\rep)$,
		where 
		\begin{itemize}
			\item 
			$ \mathit{ndet}_1(\rep)= \exists t.\dynsubcl, t'.\dynsubcl' : t.\dynsubcl\neq t'.\dynsubcl'
			\land \exists p\in\hlsysplaces\,\exists X,\comset:\\ (p.X, \comset)\in\decsetfct\land t.\dynsubcl, t'.\dynsubcl'\in\comset \land 
			\rep[t,\dynsubcl\rangle \land \rep[t'.\dynsubcl'\rangle$.
			\item $ \mathit{ndet}_2(\rep)=\exists t.\dynsubcl\, \exists p\in\hlsysplaces\,\exists X,\comset:\\ (p.X, \comset) \in\decsetfct\land t.\dynsubcl\in\comset \land \exists \dynsubcl_i^{j,k}\in\dynsubcl : |\dynsubcl_i^{j}|>1 \land 
			\rep[t,\dynsubcl\rangle  $.
		\end{itemize}
	\end{itemize}
\end{definition}
By applying \refTheo{relBetweenCanonReps}
to the properties of individual decision sets we get the following result. 
\begin{corollary}\label{cor:PropsOfSDAndCR}
	The properties of a symbolic decision set and 
	its canonical representation coincide.	
\end{corollary}

\hypertarget{target:DirStratGen}{}%
\section{Applying Canonical Representations}\label{sec:DirStratGen}
In this section, we define for a high-level Petri game \(\SG\) the two-player B\"uchi game~$ \STPGc $ with 
canonical representations~$ \rep $ of 
symbolic decision sets,
rather than arbitrary representative $ \overline{\decset}$ as in \cite{DBLP:journals/acta/GiesekingOW20}.
The edges between nodes
are directly implied by their properties and the relations 
$ \rep[t.\dynsubcl^{[\instancesubcl,\instanceelement]}\rangle\rep' $ and $ \rep[\top\rangle\rep' $ 
between canonical representations.
This allows to
\emph{directly} generate a winning strategy
for the system players in $ \PG $ from a winning strategy for player~0
in $ \STPGc $ (cf.~\refFig{HLandLLDiagram}). 

Recall that we consider high-level Petri games $ \SG\in\SGclass $, i.e.,
the represented safe P/T Petri game~\(\HLtoLL(\SG)\) 
has at most one environment player,
a bounded number of system players with a safety objective,
and the system cannot proceed infinitely without the environment.

\subsection{The Symbolic Two-Player Game}
We reduce a Petri game $ \PG $
with high-level representation $ \SG $
to a two-player B\"uchi game $ \STPGc $.
The goal is to directly create a strategy $ \PGStrat $
for the system players in $ \PG $
from a strategy $ \STPGStrat $ for player 0 in~$ \STPGc $.
Recall that a Petri game strategy must be deadlock-free, deterministic, and 
must satisfy the justified refusal condition.
Additionally, to be winning, it must not contain bad places.

The nodes in $ \STPGc $ are canonical representations of symbolic decision sets, i.e.,
they represent equivalence classes of situations in the Petri game.
The properties of canonical representations defined in \refSection{RepProps}
characterize these situations.
These properties are used in the construction of the game.
As in \cite{DBLP:journals/iandc/FinkbeinerO17,DBLP:journals/acta/GiesekingOW20}, 
the environment in the game $ \STPGc $ only moves when 
the system players cannot continue alone.
Thus, they get informed 
of the environment's decisions in their next steps 
and the system can therefore be considered as completely informed.
Bad situations (nondeterminism, deadlocks, tokens on bad places)
result in player~0 not winning.
If player~0 can avoid these situations and always win the game,
this strategy can be translated into a winning strategy for the system players in the Petri~game.
\begin{definition}[Symbolic Two-Player B\"uchi Game with Canonical Representations]
	For a high-level Petri game $ \SG $,
	$ \STPGc=( \sysnodes,\envnodes,\bgedges,\acceptingstates,\rep_0 ) $ 
	has the following components.
	\begin{itemize}
		\item The \emph{nodes} $ \bgnodes=\sysnodes\dot\cup\envnodes $
		are all possible canonical representations $ \rep $ of 
		symbolic decision sets in $ \SG $.
		The partition into \emph{player~$ 0 $'s nodes} $ \sysnodes $ 
		and \emph{player~$ 1 $'s nodes}~$ \envnodes $
		is given by $ \envnodes=\{
		\rep\with\rep $  is environment dependent$ \} $
		and $ \sysnodes=\bgnodes\setminus\envnodes $.
		\item The \emph{edges}~$ \bgedges $ are constructed as follows.
		If $ \rep\in\bgnodes $ contains a bad place, 
		is a deadlock, is terminating, 
		or is nondeterministic, there is only
		a self-loop originating from $ \rep $. 
		If $ \rep\in\sysnodes $ 
		then $ (\rep,\rep')\in\bgedges $ if either 
		$ \rep[\top\rangle\rep' $, or,
		if no $ \top $ can be resolved,
		$ \rep[t.\dynsubcl^{[\instancesubcl,\instanceelement]} \rangle \rep' $
		with only system players participating in~$ t.\dynsubcl^{[\instancesubcl,\instanceelement]} $.
		If $ \rep\in\envnodes $, then 
		$ (\rep,\rep')\in\bgedges $ for every $ \rep' $
		such that
		$ \rep[t.\dynsubcl^{[\instancesubcl,\instanceelement]} \rangle \rep' $,
		i.e., transitions involving environment players
		can only fire if nothing else is possible.
		\item The set~$ \acceptingstates $ of \emph{accepting nodes} 
		contains all representations~$ \rep $ that are 
		terminating or en\-vi\-ron\-ment-dependent, 
		but are not a deadlock, nondeterministic, or contain a bad place.
		\item The \emph{initial state} $ \rep_0 $
		is the canonical representation of the 
		symbolic decision set containing $ \decset_0 $.
	\end{itemize}
\end{definition}

A function
$ \STPGStrat:{\bgnodes}^*\sysnodes\to\bgnodes  $
s.t.
$ \forall \rep_0'\cdots\rep_k'\in
{\bgnodes}^*\sysnodes: (\rep_k',f(\rep_0'\cdots\rep_k'))\in\bgedges $
is called a \emph{strategy} for player~0.
A strategy $ \STPGStrat $ is called \emph{winning} iff
every \emph{run}
$ \rho=\rep_0 \rep_1 \rep_2\cdots $
from $ \rep_0 $ in~$ \STPGc $ (i.e., $ \forall k:
(\rep_k,\rep_{k+1})\in\bgedges $)
that is \emph{consistent with} $ \STPGStrat $
(i.e., $ \rep_k\in\sysnodes
\Rightarrow \rep_{k+1}=
\STPGStrat(\rep_0\cdots\rep_k) $)
\emph{satisfies the B\"uchi condition
	w.r.t.~$ \acceptingstates $} 
(i.e., $ \forall k\ \exists k'\geq k:
\rep_{k'}\in\acceptingstates $).

In the game~$ \STPG $ in \cite{DBLP:journals/acta/GiesekingOW20},
player~0 has a winning strategy if and only if the system players in $ \HLtoLL(\SG) $
have a winning strategy. 
As described above, it is built from the relations $ \overline{\decset}[t.c\rangle\decset' $
and $ \overline{\decset}[\top\rangle\decset' $ from representatives $ \overline{\decset} $
of symbolic decision sets.
The introduced game~$ \STPGc $ is built analogously, with the difference that 
the nodes are now canonical representations
instead of arbitrary representatives of symbolic decision sets,
and the edges are built from the relations $ \rep [t.\dynsubcl^{[\instancesubcl,\instanceelement]}\rangle \rep' $ and 
$ \rep [\top \rangle \rep' $ (cf. \refTheo{relBetweenCanonReps} and \refCor{PropsOfSDAndCR}).
The two games are isomorphic, as depicted in \refFig{HLandLLDiagram}.
Thus, we get the following result.
\begin{theorem}\label{theo:StrategyEquivalence}
	Given a high-level Petri game~$ \SG\in\SGclass $,
	there is a winning strategy for the system players in 
	the P/T Petri game
	$ \HLtoLL(\SG) $
	if and only if 
	there is a winning strategy for player~$ 0 $ in~$ \STPGc $.
\end{theorem}

The size of $ \STPGc $ is the same as of $ \STPG $ 
(exponential in the size of $ \PG $).
This means, using~$ \STPGc $, the question
whether a winning strategy in $ \PG $ exists
can still be answered in single exponential time~\cite{DBLP:journals/iandc/FinkbeinerO17}.
In $ \STPG $ we must, for a newly reached node $ \decset' $, test if 
it is equivalent to another, already existing, representative.
This means we check for all symmetries $ s\in\symmetries $
if $ s(\decset') $ is already a node in the game. 
In the best case, if we directly find the node, this is $ 1 $ comparison.
In the worst case, at step~$ i $ with currently $ |\bgnodes^i| $ nodes,
we must make $ |\symmetries_\SN||\bgnodes^i| $ comparisons 
(no symmetric node is in the game so far).
To get on the other hand the canonical representation
of a reached node in $ \STPGc $, we must make less than $ {|\symmetries_\SN|} $ 
comparisons to order the dynamic representation
(cf. \refCor{CostOfCanon}), 
and then compare it to all existing nodes.
Thus, $ |\symmetries_\SN|+1 $ comparisons in the best case
vs. $ |\symmetries_\SN|+|\bgnodes^i| $ in the $ i $-th step in the worst case.
We further investigate experimentally on this in \refSection{ExperimentalResults}.

\subsection{Direct Strategy Generation}\label{sec-directStrategyGeneration}
The solving algorithm in \cite{DBLP:journals/acta/GiesekingOW20} 
builds a strategy 
in the Petri game $ \PG=\HLtoLL(\SG) $ 
from a strategy in $ \STPG $ by 
first generating a strategy in the low-level equivalent~$ \TPG $.
Constituting the canonical representations as 
nodes 
allows us to directly generate a winning strategy~$ \PGS $ for the system players in $ \PG $
from a winning strategy $ \STPGStrat $ 
for player~$ 0 $ in $ \STPGc $ (cf. \refFig{HLandLLDiagram}).

\subsubsection{The Translation Algorithm}\label{sec:translationAlgorithm}
The key idea is the same as in \cite{DBLP:journals/iandc/FinkbeinerO17}.
The strategy $ \STPGStrat $ 
is interpreted 
as a tree~$ \strategyTree $
with labels in~$ \bgnodes $, and root~$ \treenode_0 $ 
labeled with~$ \rep_0 $.
The tree is traversed in breadth-first order,
while the strategy $ \PGS $ is extended
with every reached node.
To show that this procedure is correct,
we must show that the generated strategy $ \PGStrat $
is satisfying the conditions justified refusal,
determinism, and deadlock freedom.
Justified refusal is satisfied because of the delay of environment transitions.
Assuming nondeterminism or deadlocks in the generated strategy $ \sigma $ leads to 
the contradiction that there are respective decision sets in $ \strategyTree $.
Finally, $ \PGStrat $ is \emph{winning},
since $ \STPGStrat $ also does not reach
representations that contain a bad place.
\begin{figure}[!tb]
	\centering
	\input{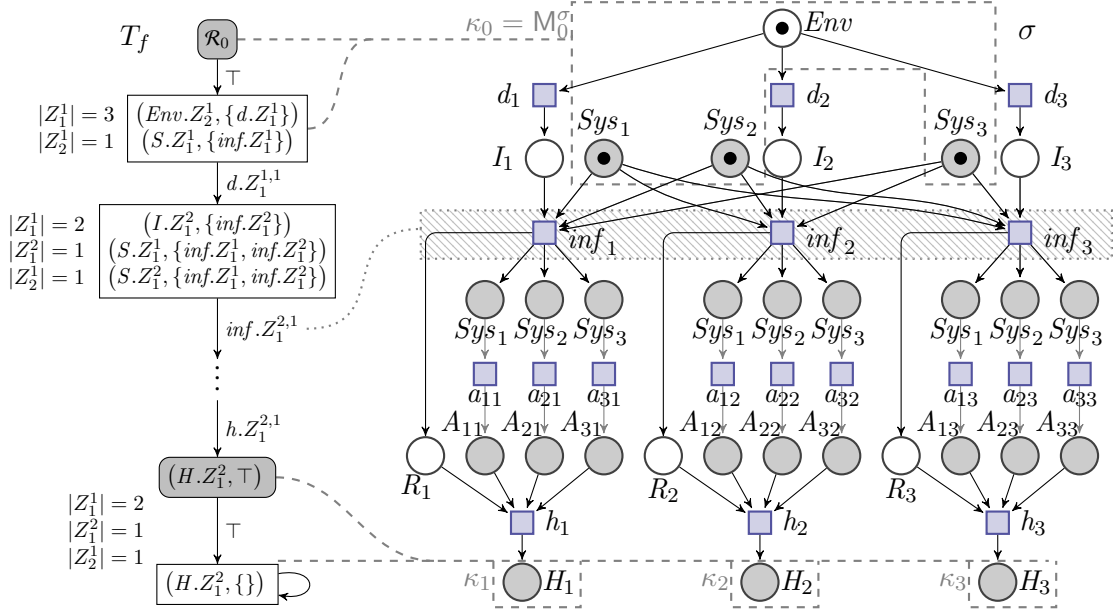}%
	\caption{\label{fig:FromFtoSigmaSmall}
		Parts of a winning strategy 
		for player~$ 0 $ in
		$ \STPGc $ (a tree with gray nodes for player~$ 0 $),
		and the generated strategy
		for the system players in $ \HLtoLL(\SG) $.
	}
\end{figure}

Now we describe the algorithm in detail.
Initially, the strategy $ \PGS $
contains places corresponding to the 
initial marking $ \marking_0 $ 
in the Petri game, i.e.,
places labeled with $ p.c $
for every $ p.c\in\marking_0 $, 
each with a token on them.
They constitute the initial marking of $ \PGStrat $.
Every node $ \treenode $ in $ \strategyTree $,
labeled with $ \rep=(\dynsubclasses,\decsetfct) $,
is now associated to a set $ \cuts_{\treenode} $ of \emph{cuts} --
reachable markings in the strategy/unfolding.
The set $ \cuts_{\treenode_0} $,
associated to the root $ \treenode_0 $, contains only
the cut $ \cut_0 $ consisting of 
the places in the 
initial marking described above.
Every cut $ \cut $ is equipped with
a valid assignment~$ \validassignment_\cut $
that maps $ \cut $ to 
places in the dynamic decision set
$ \decsetfct $.
Every edge~$ (\treenode,\treenode') $ in $ \strategyTree $ 
corresponds to an edge~$ (\rep,\rep') $
in $ \STPGc $, and the edges there again can correspond to 
either relations $ \rep[t.\dynsubcl^{[\instancesubcl,
	\instanceelement]}\rangle \rep' $
or $ \rep[\top\rangle\rep' $.

Suppose in the breadth-first traversal of $ \strategyTree $
we reach a node~$ \treenode $ with label $ \rep $
and set of associated cuts $ \cuts_\treenode $. 
Further suppose there is an edge in $ \strategyTree $
from $ \treenode $ to a node $ \treenode' $ labeled with $ \rep' $.
From there, we distinguish three cases:

1) The edge $ (\treenode,\treenode') $ in $ \strategyTree $
corresponds to a relation $ \rep[\top\rangle\rep' $ in $ \SG $. 
As we know, this can be depicted as 
$ \rep\xrightarrow{\text{splitting}} 
\rep_s \xrightarrow{\top\text{-resolution}}
\rep_s'\xrightarrow{\text{canon. rep. } }\rep'$.
We now describe how $ \cuts_\treenode $ changes in these three steps.
In the first step, the partition $ \rep_s $, there is a
family $ \partitionfunc $ of functions 
$ \partitionfunc_i:({\dynsubclasses_s})_i
\to \dynsubclasses_i $
describing the splitting.
For every $ (\cut,\validassignment_{\cut})\in \cuts_{\treenode} $,
we change the assignment to a function
$ \validassignment_{\cut}':\bigcup_i\basecl_i\to \dynsubclasses_s $
such that 
$ \validassignment_{\cut}=\partitionfunc
\circ\validassignment_{\cut}' $.
In the second step, the resolution of $ \top $,
the set $ \cuts_\treenode $ does not change 
(and
we now denote $ \validassignment_{\cut}':
\bigcup_i\basecl_i\to\dynsubclasses_s' $) .
For the third step (the canonical representation)
there again is a partition function
$ \partitionfunc': \dynsubclasses_s'\to
\dynsubclasses' $. 
We again change the assignment
for every cut $ \cut $ to
$ \validassignment_\cut''=
\partitionfunc'\circ\validassignment_\cut' $.
Finally we assign to node $ \treenode' $ the set
$ \cuts_{\treenode'}=\{ (\cut,\validassignment_\cut'')\with
(\cut,\validassignment_{\cut})\in\cuts_\treenode \} $.

2) The edge $ (\treenode,\treenode') $ in $ \strategyTree $
corresponds to a relation $ \rep[
t.\dynsubcl^{[\instancesubcl,\instanceelement]} \rangle\rep' $ 
in $ \SG $ and $ t $ is an system transition. 
The relation can be described by 
$ \rep\xrightarrow{\text{splitting w.r.t } [\instancesubcl,\instanceelement]} 
\rep_s \xrightarrow{\text{firing } t.\dynsubcl^{[\instancesubcl,\instanceelement]}}
\rep_s'\xrightarrow{\text{canon. rep. } }\rep'$.
Since the first step (splitting) is a special case of partition,
we proceed as in the first case, and get the same cuts~$ \cut $
with altered assignments
$ \validassignment_\cut': 
\bigcup_i\basecl_i\to \dynsubclasses_s $.
Since after the process of splitting, 
the dynamic subclasses appearing in 
$ \dynsubcl^{[\instancesubcl,\instanceelement]} $
are all of cardinality $ 1 $, there is for each cut $ \cut $
only one mode $ v_\cut=((c_{i,j})_{j=1}^{t_i})_{i=1}^n\in\valuations(t) $
such that $ \validassignment_\cut(v_\cut)=
\dynsubcl^{[\instancesubcl,\instanceelement]} $.
For every cut $ \cut $ we add a transition 
with label $ t.d $ to the strategy, 
the places labeled with $ p.c $ from $ \cut $
in its preset such that $ c\in v(p,t) $,
and we add places labeled with~$ p'.c' $ to the 
strategy in the postset of the transition if 
$ c'\in v(t,p') $. 
We then delete the places in the preset 
of the new transition from $ \cut $ and add the places 
in the postset, to get the cut~$ \cut' $ with 
the same assignment $ \validassignment_{\cut'}
=\validassignment_\cut':
\bigcup_i\basecl_i\to \dynsubclasses_s' $.
In the third step (canonical representation) we can
again proceed as in the first case, which results in 
$ \cuts_{\treenode'}=\{ (\cut',\validassignment_\cut'')\with
(\cut,\validassignment_{\cut})\in\cuts_\treenode \} $.

3) The edge $ (\treenode,\treenode') $ in $ \strategyTree $
corresponds to a relation $ \rep[
t.\dynsubcl^{[\instancesubcl,\instanceelement]} \rangle\rep' $ 
in $ \SG $ and $ t $ is an environment transition. 
In this case, we proceed exactly as in the second case,
but with one crucial change in the first step 
(splitting w.r.t $ [\instancesubcl,\instanceelement] $):
instead of choosing \emph{one} function
$ \validassignment_{\cut}':\bigcup_i\basecl_i\to \bigcup_i\rep_p.\symbcl_i $
such that 
$ \validassignment_{\cut}=\partitionfunc
\circ\validassignment_{\cut}' $,
we add a pair $ (\cut,\validassignment_\cut') $ for 
\emph{all} such assignments and all corresponding 
$ t.\dynsubcl^{[\instancesubcl,\instanceelement]} $.
The rest of the procedure remains the same.

In \refFig{FromFtoSigmaSmall}, the strategy tree
$ \strategyTree $ (consisting of only one branch) in $ \STPGc $ and the generated strategy~$ \PGStrat $ in 
(the unfolding of) $ \HLtoLL(\SG) $ for the running example $ \SG $ are depicted.
The strategy $ \PGStrat $ was already informally described in \refSection{Introduction}
and partly shown in \refFig{UnfAndStrat}.
The initial canonical representation $ \rep_0 $
is associated to the cut representing the initial marking in the Petri game.
The $ \top $-resolution does not change the associated cuts.
Firing $ \mathit{d}.\dynsubcl_1^{2,1} $ corresponds to the three firings 
of $ \mathit{d}_1 $, $ \mathit{d}_2 $, and $ \mathit{d}_3 $ in the strategy.
Thus, the third canonical representation is associated to the three cuts
$ \{ \mathit{Sys}_1,\mathit{Sys}_2,\mathit{Sys}_3, \mathit{I}_j \} $, $ j=1,2,3 $.
The strategy $ \strategyTree $ terminates in the canonical representation at the bottom,
which corresponds to the three situations where all computers connected to the correct host.

\subsubsection{Correctness of the Algorithm}
Finally, we prove that
the algorithm 
is correct.
This proof works analogous to 
\cite{DBLP:journals/iandc/FinkbeinerO17},
where Finkbeiner and Olderog consider the case of P/T Petri games.
The crucial difference is that we associate multiple 
cuts in $ \PGStrat $ to a node in the tree~$ \strategyTree $.
We have to prove that the generated strategy $ \PGStrat $
is satisfying the conditions justified refusal,
determinism and deadlock freedom defined above.
Then, $ \PGStrat $ is \emph{winning},
since $ \STPGStrat $ does not reach
representations that contain a bad place.

\begin{theorem}
	Given a high-level Petri game $ \SG\in\SGclass $,
	the algorithm presented in \refSection{translationAlgorithm} translates a winning strategy for player~0 in $ \STPGc $ 
	into a winning strategy for the system players in~$ \HLtoLL(\SG) $.
\end{theorem}

\begin{proof}
	\emph{Justified refusal:}
	Assume there is a transition $ \transition $ in the unfolding
	that is not in the strategy, but the preset
	of the transition is (so $ \transition $ would be enabled in the strategy).
	Let $ t.c=\lambda(\transition) $.
	Then, the reason for $ \transition $'s absence is that 
	there is a system place $ \place $ in the preset of $ \transition $
	such that on all traversals that reach a representation $ \rep=(\dynsubclasses,\decsetfct) $
	and a cut with valid assignment $ (\cut,\validassignment_\cut) $
	such that $ \transition $ is enabled in $ \cut $,
	$ \decsetfct $ forbids $ \validassignment_\cut(\lambda(\transition))=t.\validassignment_\cut(c) $,
	and hence also no transition $ \mathsf{\transition'} $ 
	with $ \lambda(\mathsf{t'})=\lambda(\mathsf{t}) $
	(implying $ \validassignment_\cut(\lambda(\transition))=\validassignment_\cut(\lambda(\mathsf{\transition'})) $)
	is in the strategy.
		
	\emph{Deadlock freedom:}
	The generated strategy is deadlock free because
	in every branch in the tree~$ \strategyTree $ 
	either infinitely many transitions are fired,
	or at some point it loops
	in a terminating state.
	The strategy~$ \STPGStrat $ does not contain representations that are deadlocks.

	\emph{Determinism:}
	Aiming a contradiction, assume that there
	are two transitions $ \mathsf{t_1} $ and $ \mathsf{t_2} $
	in the strategy that have a shared system place $ \place $ in their preset.
	and are both enabled at the same cut $ \cut $ in the strategy.\\
	\emph{Case 1.}
	Suppose that during the construction of $ \PGStrat $,
	$ \cut $ is encountered as a cut $ (\cut,\validassignment_\cut)\in\cuts_{\treenode} $
	associated to node $ \treenode $ labeled with representation $ \rep=(\dynsubclasses,\decsetfct) $.
	Then $ \rep $ is nondeterministic. Contradiction (the strategy does not contain nondeterministic representations).\\
	\emph{Case 2.}
	Suppose there exist two cuts $ (\cut_1,\validassignment_{\cut_1})\in\cuts_{\treenode_1} $
	and $ (\cut_2,\validassignment_{\cut_2})\in\cuts_{\treenode_2} $
	such that during the construction of $ \PGStrat $,
	first $ \mathsf{t_1} $ was introduced at $ \cut_1 $
	and then $ \mathsf{t_2} $ was introduced at $ \cut_2 $,
	with $ \preset{\mathsf{t_i}}\subseteq\cut_i $ but $ \preset{\mathsf{t_1}}\nsubseteq\cut_2 $
	and $ \preset{\mathsf{t_2}}\nsubseteq\cut_1 $.
	We again distinguish two cases.\\
	\emph{Case 2a} concerns the situation that 
	$ \treenode_2 $ is \emph{below} $ \treenode_1 $ in $ \strategyTree $.
	Suppose in the B\"uchi game the symbolic instances 
	$ t_1.\dynsubcl^{[\instancesubcl^1,\instanceelement^1]},\dots,t_n.\dynsubcl^{[\instancesubcl^n,\instanceelement^n]} $,
	with $ n\geq 1 $, are fired between $ \treenode_1 $ and $ \treenode_2 $,
	corresponding to a sequence 
	$ \cut_1=\cut_0'[\mathsf{\transition_1'}\rangle\dots[\mathsf{\transition_n'}\rangle\cut_n'=\cut_2 $.
	If this sequence removes a token from a place $ \place\in\preset{\mathsf{\transition_1}} $
	and puts it on $ \preset{\mathsf{\transition_2}} $ then place $ \preset{\mathsf{\transition_1}} $ and $ \preset{\mathsf{\transition_2}} $ are in conflict, therefore there cannot exist a cut~$ \cut $ in which both $ \mathsf{\transition_1} $ and $ \mathsf{\transition_2} $ are enabled. Contradiction.
	Thus, a transition that takes a token from a place in~$ \preset{\mathsf{\transition_1}} $
	puts it outside of $ \preset{\mathsf{\transition_2}} $.
	Consider the first transition $ \mathsf{\transition_j} $ that removes a token from $ \preset{\mathsf{\transition_1}} $.
	Then the representation $ \rep=(\dynsubclasses,\decsetfct) $, 
	label of the node $ \treenode $ that introduces $ \cut_{j-1} $
	is nondeterministic, as two instances 
	$ t.\dynsubcl^{[\instancesubcl,\instanceelement]} $ and $ t_j.\dynsubcl^{[\instancesubcl^i,\instanceelement^i]} $,
	corresponding to 
	$ \mathsf{\transition_1} $ and $ \mathsf{\transition_j} $
	are enabled in $ \decsetfct $. 
	Then proceed as in~Case~1.\\
	\emph{Case 2b} concerns the situation that $ \treenode_1 $ and $ \treenode_2 $
	are on different branches in $ \strategyTree $. 
	Let $ \treenode $ be their last common ancestor.
	Let $ \place $ be the environment place contained in the 
	cut $ (\cut,\validassignment_\cut)\in\cuts_\treenode $
	that such that $ \cut\leq\cut_i $.
	Let $ \mathsf{\transition_i'} $ be the transition 
	from $ \place $ in the path from $ \treenode $ to $ \treenode_i $.
	Then all transitions added in branches after $ \mathsf{t_1'} $
	are in conflict with all transitions added in branches after $ \mathsf{t_2'} $.
	In particular, $ \mathsf{t_1} $ and $ \mathsf{t_2} $ are in conflict,
	and therefore cannot be enabled at the same marking. Contradiction.
\end{proof}

\section{Experimental Results}\label{sec:ExperimentalResults}
In this section we investigate the impact of using \emph{canonical representations} for solving the realizability problem
of distributed systems modeled with
high-level Petri games with one environment player, an arbitrary number of system players, a safety objective, and an underlying symmetric net.

A prototype \cite{DBLP:journals/acta/GiesekingOW20} for generating the reduced state space of $ \STPG $ for
such a high-level Petri game $ \SG $ 
shows a state space reduction by up to three orders of magnitude
compared to $ \mathcal{G}(\HLtoLL(\SG)) $ (cf.~\refFig{HLandLLDiagram}) for the considered benchmark families~\cite{DBLP:journals/acta/GiesekingOW20}.
For this paper we extended this prototype and implemented algorithms
to obtain the same state space reduction by using canonical representations in~/$ \STPGc $.
Furthermore, we implemented a solving algorithm to exploit the reduced state space for the realizability problem
of high-level Petri games.
As a reference, we implemented an explicit approach which does not exploit any symmetries of the system.
We applied our algorithms on the benchmark families presented in~\cite{DBLP:journals/acta/GiesekingOW20} and
added a benchmark family for the running example introduced in this paper.
An extract of the results for three of these benchmark families are given in Table~\ref{tab:results}. The complete 
results are contained in the corresponding artifact~\cite{Gieseking2021}.
\begin{table}[hbt]
	\caption{
		Comparison of the run times of the \emph{canonical (Canon.)} and \emph{membership (Memb.)} approach solving
		the realizability problem (\cmark/\xmark) for the 3 benchmark families \emph{CS}, \emph{DW}, \emph{CM}
		with the number of states \(|V|\) and number of symmetries~\(|\symmetries|\).
		A gray number of states \(|\mathsf{V}|\) for the \emph{explicit reference approach} indicates a timeout.
		Results are obtained on an AMD Ryzen\texttrademark{} 7 3700X~CPU, 4.4~GHz, 64 GB RAM
		and a timeout (TO) of 30 minutes. The run times are in seconds.}
	\label{tab:results}
	\begin{minipage}{.495\textwidth}
		\scalebox{0.93}{
			\centering
			{\setlength{\tabcolsep}{1mm} 	  
			\begin{tabular}{c||r||c||rr||rr}
				\textit{CS} & \(|\mathsf{V}|\)  & \(\models\) & \(|V|\) & \(|\symmetries|\) & Memb. & Canon. \\\hline
				1 & 21 &  \cmark & 21 & 1 & .38 & \textbf{.36} \\
				2 & 639 &  \cmark & 326 & 2 & \textbf{.63} & .64 \\
				3 & 45042 &  \cmark & 7738 &  6 & \textbf{5.20} & 6.05 \\
				4 & \lightText{7.225e6} &  \cmark & 3.100e5 &  24 & 151.62 & \textbf{148.08} \\
				5 & \lightText{3.154e9} &   - & - & 120 & TO & TO \\
				\multicolumn{7}{c}{}\\
				\textit{DW} & \(|\mathsf{V}|\)  & \(\models\) & \(|V|\) & \(|\symmetries|\) & Memb. & Canon. \\\hline
				1 & 57 & \cmark & 57 & 1 & .40 & \textbf{.39} \\
				2 & 457 & \cmark & 241 & 2 & .67 & \textbf{.62} \\
				& \multicolumn{1}{|c}{\(\cdots\)}&\multicolumn{1}{c}{\(\cdots\)}&\multicolumn{2}{c}{\(\cdots\)}&\multicolumn{2}{c}{\(\cdots\)}\\
				7 & \lightText{4.055e6} & \cmark & 5.793e5 & 7 & 100.67 & \textbf{75.24} \\
				8 & \lightText{2.097e7} & \cmark & 2.621e6 & 8 & 986.77 & \textbf{671.04} \\
				9 & \lightText{1.053e8} & - & - &  9 & TO & TO 
			\end{tabular}
			}
		}
	\end{minipage}\hfill
	\begin{minipage}{.505\textwidth}
		\scalebox{0.93}{
			\centering
			{\setlength{\tabcolsep}{1mm} 	  
			\begin{tabular}{c||r||c||rr||rr}
				\textit{CM} & \(|\mathsf{V}|\)  & \(\models\) & \(|V|\) & \(|\symmetries|\) & Memb. & Canon. \\\hline
				2/1 & 155 & \cmark & 79 & 2 & \textbf{.49} & .52 \\
				2/2 & 2883 & \xmark & 760 & 4 & \textbf{1.07} & 1.08 \\
				2/3 & 58501 & \xmark & 5548 &  12 & \textbf{4.38} & 5.94 \\
				2/4 & \lightText{1.437e6} & \xmark & 33250 & 48 & 15.12 & \textbf{14.40} \\
				2/5 & \lightText{3.419e7} & \xmark & 1.701e5 & 240 & 296.05 & \textbf{185.81} \\
				2/6 & \lightText{8.376e8} &  - & - & 1440 & TO & TO \\\cdashline{2-7}
				3/1 & 702 & \cmark & 147 & 6 & .71 & \textbf{.58} \\
				3/2 & 45071 & \cmark & 4048 & 12 & \textbf{4.46} & 4.99 \\
				3/3 & \lightText{3.431e6} & \xmark & 91817 & 36 & 89.35 & \textbf{49.90} \\
				3/4 & \lightText{2.622e8} & - & - & 144 & TO & TO \\\cdashline{2-7}
				4/1 & 2917 & \cmark & 239 & 24 & \textbf{1.24} & 1.42 \\
				4/2 & \lightText{6.587e5} & \cmark & 16012 & 48 & 25.42 & \textbf{14.09} \\
				4/3 & \lightText{1.546e8} & - & - & 144 & TO & TO \\
			\end{tabular}
			}
		}
	\end{minipage}
\end{table}

The benchmark family \textit{Client/Server (CS)} corresponds to the running example of the paper.
With \textit{Document Workflow (DW)} a cyclic document workflow between clerks is modeled. In this 
benchmark family the symmetries of the systems are only one rotation per clerk.
In \textit{Concurrent Machines (CM)} a hostile environment can destroy one of the machines
processing the orders. Since each machine can only process one order, a positive realizability result
is only obtained when the number of orders is smaller than the number of machines.
In Table~\ref{tab:results} we can see that for those benchmark families
the extra effort of computing the canonical representations 
(Canon.)
is worthwhile for most instances
compared to the cost of checking the membership of a decision set in an equivalence class~(Memb.).
This is not the case for all benchmark families.
\begin{figure}[tb]
	\centering
	\includegraphics[width=\textwidth]{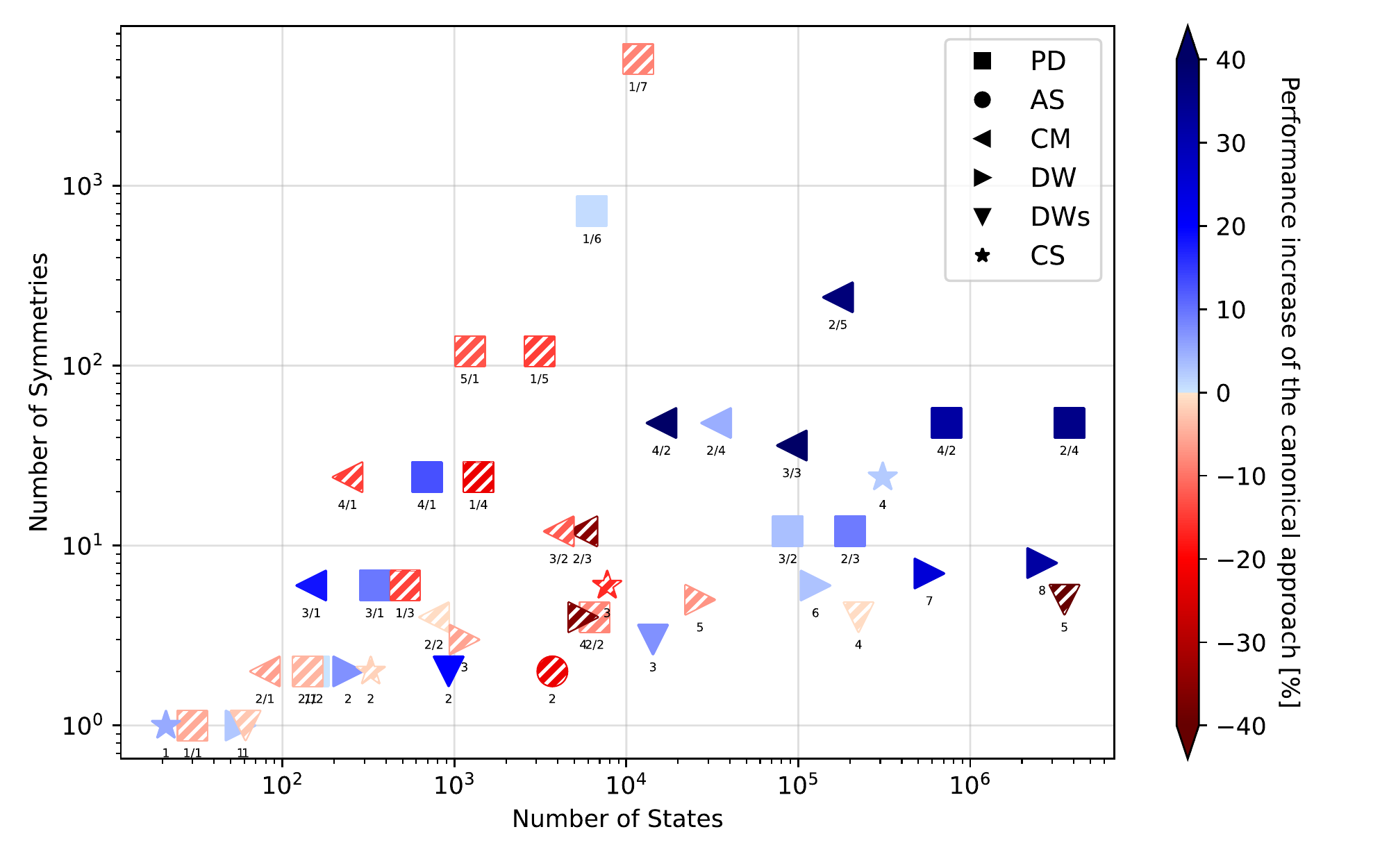}%
	\captionof{figure}{Comparing the percentage performance gain of the canonical and the membership approach with respect to
		the number of states and symmetries of the input problem for the benchmark families
		\emph{Package Delivery (PD)}, \emph{Alarm System (AS)}, \emph{CM}, \emph{DW}, \emph{DWs}, \emph{CS}.
		Labels are the parameters of the benchmark. A blue (unhatched) marker indicates a performance increase when using canonical representations.}
	\label{fig:canonVsMember}
\end{figure}

In \refFig{canonVsMember} we have plotted the instances of all benchmark families
according to their number of symmetries and states.
The color of the marker shows the percentaged in- or decrease in performance
when using canonical representations while solving high-level Petri games.
Blue (unhatched) indicates a performance gain when using the canonical approach.
This shows that the benchmarks in general benefit from the canonical approach
for an increasing number of states (the right blue (unhatched) area).
However, the \emph{DWs} benchmark (a simplified version of \emph{DW}) exhibits the opposite behavior.
This is most likely explained by the very simple structure, which favors a quick member check.

The algorithms are integrated in \textsc{AdamSYNT}\footnote{\url{https://github.com/adamtool/adamsynt}}
\cite{DBLP:conf/cav/FinkbeinerGO15,DBLP:journals/corr/abs-1711-10637}, open source, 
and available online\footnote{\url{https://github.com/adamtool/highlevel}}. 
Additionally, we created an artifact with the current version running in a virtual machine
for reproducing and checking all experimental data with provided scripts \cite{Gieseking2021}.

\section{Related Work}\label{sec:RelatedWork}%
For the synthesis of distributed systems other approaches are most prominently the Pnueli/Ros\-ner model
\cite{DBLP:conf/focs/PnueliR90} and Zielonka's asynchronous automata~\cite{DBLP:books/ws/95/Zielonka95}.
The synchronous setting of Pnueli/Rosner is in general undecidable \cite{DBLP:conf/focs/PnueliR90},
but some interesting architectures exist that have a decision procedure with nonelementary complexity \cite{Rosner/92/Modular,DBLP:conf/lics/KupfermanV01,DBLP:conf/lics/FinkbeinerS05}.
For asynchronous automata, the decidability of the control problem is open in general,
but again there are several interesting cases which have a decision procedure with nonelementary complexity
\cite{DBLP:conf/icalp/GenestGMW13,DBLP:conf/fsttcs/MuschollW14,DBLP:conf/fsttcs/MadhusudanTY05}. 

Petri games based on P/T Petri nets are introduced in \cite{DBLP:journals/corr/FinkbeinerO14,DBLP:journals/iandc/FinkbeinerO17}.
Solving unbounded Petri games is in general undecidable.
However, for Petri games with one environment player, a bounded number of system players,
and a safety objective the problem is \textsc{exptime}-complete.
The same complexity result holds for interchanged players~\cite{DBLP:conf/fsttcs/FinkbeinerG17}.
High-level Petri games have been introduced in \cite{DBLP:journals/corr/abs-1904-05621}.
In \cite{DBLP:journals/acta/GiesekingOW20}, such Petri games are solved while exploiting symmetries. 
\cite{DBLP:journals/it/Wurdemann21} gives outlooks towards high-level representations of strategies.

The symbolic B\"uchi game is inspired by the symbolic reachability graph for high-level nets from~\cite{DBLP:journals/tcs/ChiolaDFH97},
and the calculation of canonical representatives~\cite{Chiola1991} from~\cite{DBLP:journals/tc/ChiolaDFH93}. 
There are several works on how to obtain symmetries of different subclasses of high-level Petri nets efficiently
\cite{DBLP:conf/apn/DutheilletH89,%
	Chiola1991,%
	DBLP:journals/tc/ChiolaDFH93,%
	DBLP:conf/apn/Lindquist91%
}
and for efficiency improvements for systems with different degrees of symmetrical behavior
\cite{DBLP:conf/apn/HaddadITZ95,%
	BAARIR2004219,%
	DBLP:conf/mascots/BellettiniC04%
}.

\section{Conclusions and Outlook}\label{sec:Conclusions}
We presented a new construction for the synthesis of distributed systems modeled by
high-level Petri games with one environment player, an arbitrary number of system players,
and a safety objective.
The main idea is the reduction to a 
symbolic two-player B\"uchi game,
in which the nodes are equivalence classes 
of symmetric situations in the Petri game.
This leads to a significant reduction of the state space.
The novelty of this construction is to obtain the reduction
by introducing canonical representations. 
To this end, a theoretically cheaper construction of the B\"uchi game
can be obtained depending on the input system.
Additionally, the representations now allow to skip 
the inflated generation of an explicit B\"uchi game strategy
and to directly generate a Petri game strategy 
from the symbolic B\"uchi game strategy.
Our implementation, applied on six structurally different benchmark families,
shows in general a performance gain in favor of the canonical representatives
for larger state spaces.

In future work, we plan to integrate the algorithms in AdamWEB \cite{DBLP:conf/tacas/GiesekingHY21}, 
a web interface\footnote{\url{http://adam.informatik.uni-oldenburg.de/}} for the synthesis of distributed systems,
to allow for an easy insight in the symbolic games and strategies.
Furthermore, we want to continue our investigation on the benefits 
of canonical representations, e.g., to 
directly generate high-level representations of Petri game strategies that match the 
given high-level Petri~game.
\bibliographystyle{fundam}
\bibliography{bib}

\appendix
\section*{Appendix}
\section{Strategies in Petri Games -- Formal Definition}

\label{app:strategyFormal}
We formally define strategies in Petri games as subprocesses of the unfolding. 

First, we define concurrency and conflicts in Petri nets.
Consider a Petri net 
$ \PN=(\places,\transitions,\flowfunc,\marking_\mathsf{0}) $
and two nodes $ \mathsf{x},\mathsf{y}\in\places\cup\transitions $.
We write $ \mathsf{x}<\mathsf{y} $, and call $ \mathsf{x} $ \emph{causal predecessor} of $ \mathsf{y} $,
if $ \mathsf{x}\flowfunc^+\mathsf{y} $.  
Additionally, we write $ \mathsf{x}\leq \mathsf{y} $ if $ \mathsf{x}<\mathsf{y} $ or $ \mathsf{x}=\mathsf{y} $,
and call $ \mathsf{x} $ and $ \mathsf{y} $ \emph{causally related} if 
$ \mathsf{x}\leq \mathsf{y} $ or $ \mathsf{y}\leq \mathsf{x} $.
The nodes $ \mathsf{x} $ and $ \mathsf{y} $ are said to be \emph{in conflict},
denoted by $ \mathsf{x}\sharp \mathsf{y} $, if 
$ \exists \mathsf{p}\in\places\,\exists \mathsf{t_1}\neq \mathsf{t_2}\in\postset{p} :
\mathsf{t_1}\leq \mathsf{x} \land \mathsf{t_2}\leq \mathsf{y}$. 
Two nodes are called \emph{concurrent},
if they are neither in conflict, nor causally related.
A set $ \mathsf{X}\subseteq\places\cup\transitions $
is called concurrent, if each two elements in $ \mathsf{X} $ are concurrent.

$ \PN $ is called an \emph{occurrence net} if: 
i) every place has at most one transition in its preset, i.e.,
$ \forall \mathsf{p}\in\places:|\preset{\mathsf{p}}|\leq 1 $;
ii) no transition is in self conflict, i.e., 
$ \forall {\mathsf{t}\in\transitions:} \neg(\mathsf{t}\sharp \mathsf{t}) $;
iii) no node is its own causal predecessor, i.e.,
$ \forall \mathsf{x}\in\places\cup\transitions:{\neg (\mathsf{x}<\mathsf{x})} $;
iv) the initial marking contains exactly the places with no transition in their preset, i.e.,
$ \marking_\mathsf{0}=\{ \mathsf{p}\in\places\with \preset{\mathsf{p}}=\emptyset \} $.
An occurrence net is called a \emph{causal net} 
if additionally 
v) from every place there is at most one outgoing transition, i.e.,
$ \forall \mathsf{p}\in\places: |\postset{\mathsf{p}}|\leq 1 $.

For a superscripted net $ \PN^x $,
we implicitly also equip its components with 
the superscript, i.e., 
$ \PN^x=(\places^x,\transitions^x,\flowfunc^x,\marking_\mathsf{0}^x) $,
as well as the corresponding pre- and postset function 
$ \mathit{pre}^x,\mathit{post}^x $.
Let $ \PN $ and $ \PN' $ be two Petri nets.
A function $ h:\places\cup\transitions\to\places'\cup\transitions' $
is called a \emph{Petri net homomorphism}, if:
i) it maps places and transitions in $ \PN $ into the corresponding sets in $ \PN' $, i.e.,
$ \forall \mathsf{p}\in\places:{h(\mathsf{p})\in\places'}\land\forall {\mathsf{t}\in\transitions}: h(\mathsf{t})\in\transitions' $;
ii) it maps the pre- and postset correspondingly, i.e.,
$ \forall {\mathsf{t}\in\transitions}: \mathit{pre}'(h(\mathsf{t}))=h(\preset{\mathsf{t}}) 
\land \mathit{post}'(h(\mathsf{t}))=h(\postset{\mathsf{t}})$.
The homomorphism is called \emph{initial} if additionally
iii) it maps the initial marking of $ \PN $ to the initial marking of~$ \PN' $,
i.e., $ h(\marking_\mathsf{0})=\marking_\mathsf{0}' $.

A(n \emph{initial}) \emph{branching process}
$ \beta=(\PN^U,\lambda^U) $ of $ \PN $
consists of an occurrence net $ \PN^U $ and a(n initial) 
homomorphism $ \lambda^U:\places^U\cup\transitions^U
\to \places\cup\transitions $ that 
is injective on transitions with same preset, i.e.,
$ \forall \mathsf{\mathsf{t_1}},\mathsf{t_2}\in\transitions^U: 
(\preset[U]{\mathsf{t_1}}=\preset[U]{\mathsf{t_2}}\land \lambda^U(\mathsf{t_1})=\lambda^U(\mathsf{t_2}) )
\Rightarrow \mathsf{t_1}=\mathsf{t_2} $.
If $ \beta_R=(\PN^R,\rho) $ is an initial 
branching process of $ \PN $ with a causal net $ \PN^R $,
$ \beta_R $ is called an \emph{initial (concurrent) run of $ \PN $}.
A run formalizes a single concurrent execution of the net.
For two branching processes 
$ \beta=(\PN^U,\lambda^U) $ and $ \beta'=({\PN^V},{\lambda^V}) $
we call $ \beta $ a \emph{subprocess} of $ \beta' $, if
i) $ \PN^U $ is a \emph{subnet} of $ {\PN^V} $, i.e.,
$ \places^U\subseteq{\places^V},
\transitions^U\subseteq{\transitions^V},
\flowfunc^U\subseteq{\flowfunc^V},
\marking_\mathsf{0}^U={\marking_\mathsf{0}^V} $;
ii) $ {\lambda^V} $ acts on $ \places^U\cup\transitions^U $
as $ \lambda^U $ does, i.e.,
$ {\lambda^V}|_{\places^U\cup\transitions^U}=\lambda^U $.

A branching process $ \beta=(\PN^U,\lambda^U) $
is called an \emph{unfolding of $ \PN $}, if 
for every transition that can occur in the net,
there is a transition in the unfolding with corresponding label, i.e.,
$ \forall \mathsf{t}\in\transitions\,\forall X\subseteq\places^U:
X \text{ concurrent} \land \preset{\mathsf{t}}=\lambda^U(X)
\Rightarrow \exists \mathsf{t}^U\in\transitions^U:
\lambda^U(\mathsf{t}^U)=\mathsf{t} \land \preset[U]{\mathsf{t}^U}=X $.
The unfolding of a net is unique up to isomorphism.
The unfolding of a Petri game 
$ \PG=(\sysplaces,\envplaces,\transitions,\flowfunc,\marking_0,\badplaces) $ 
with underlying net $ \PN $ is the unfolding of $ \PN $,
where the distinction of system-, environment-, and bad places 
is lifted to the branching process:
$ \sysplaces^U=\{ \mathsf{p}\in\places^U\with \lambda^U(\mathsf{p})\in\sysplaces \} $,
$ \envplaces^U=\{ \mathsf{p}\in\places^U\with \lambda^U(\mathsf{p})\in\envplaces \} $,
and
$ \badplaces^U=\{ \mathsf{p}\in\places^U\with {\lambda^U(\mathsf{p})\in\badplaces} \} $.

A \emph{strategy} for the system players in $ \PG $
is a subprocess $ \PGStrat=(\PN^\PGStrat,\lambda^\PGStrat) $ 
of the unfolding of $ \PG $ satisfying the following conditions:\\
\emph{Justified refusal:}
If a transition is forbidden in the strategy,
then a system player in its preset uniformly forbids all occurrences in the strategy, i.e.,
$ \forall \mathsf{t}\in\transitions^U: ({\mathsf{t}\notin\transitions^\PGStrat} 
\land \preset[U]{t}\subseteq\places^\PGStrat)\Rightarrow
(\exists \mathsf{p}\in\preset[U]{\mathsf{t}}\cap\sysplaces^\PGStrat\,
\forall \mathsf{t'}\in\postset[U]{\mathsf{p}}: \lambda^U(\mathsf{t'})=\lambda^U(\mathsf{t})
\Rightarrow \mathsf{t'}\notin\transitions^\PGStrat)$.
This also implies that no pure environment transition is forbidden. \\
\emph{Determinism:}
In no reachable marking (cut) in the strategy does a 
system player allow two transitions in his postset 
that are both enabled, i.e.,
$ \forall \mathsf{p}\in\sysplaces^\PGStrat\ 
\forall \marking\in\mathcal{R}(\PN^\PGStrat):
\mathsf{p}\in\marking\Rightarrow \exists^{\leq 1} \mathsf{t}\in\postset[\PGStrat]{\mathsf{p}}:
\marking[\mathsf{t}\rangle $. The set $ \mathcal{R}(\PN^\PGStrat) $
are all reachable markings in the strategy.\\
\emph{Deadlock freedom:}
Whenever the system can proceed in $ \PG $, 
the strategy must also give the possibility to continue, i.e.,
$ \forall \marking\in\mathcal{R}(\PN^\PGStrat):
(\exists \mathsf{t}\in\transitions^U: \marking[\mathsf{t}\rangle ) 
\Rightarrow {\exists \mathsf{t'}\in\transitions^\PGStrat: \marking[\mathsf{t'}\rangle} $.

An initial concurrent run $ \pi=(\PN^R,\rho) $ of
the underlying net $ \PN $ of a Petri game~$ \PG $
is called a \emph{play} in~$ \PG $.
The play $ \pi $ \emph{conforms} to $ \PGStrat $
if it is a subprocess of $ \PGStrat $.
The system players win $ \pi $ if $ \badplaces^R=\emptyset $,
otherwise the environment players win $ \pi $.
The strategy $ \PGStrat $ is called \emph{winning},
if all plays that conform to $ \PGStrat $
are won by the system players.
This is equivalent to $ \badplaces^\PGStrat=\emptyset $.
\end{document}